\documentclass[12pt,authoryear,review]{elsarticle}
\makeatletter

  \def\ps@pprintTitle{%
 \let\@oddhead\@empty
 \let\@evenhead\@empty
 \def\@oddfoot{\centerline{\thepage}}%
 \let\@evenfoot\@oddfoot}
\makeatother

\usepackage{amsmath,amssymb,amsfonts,dsfont,ulem}
\usepackage{graphicx}
\usepackage{subfigure}
\usepackage{color}
\usepackage{setspace}
\usepackage{pdflscape}
\usepackage{soul}
\usepackage{multirow}
\usepackage{bigstrut}
\usepackage[rightcaption]{sidecap}

\journal{TBA}

\newtheorem{theorem}{Theorem}

\newtheorem{proposition}[theorem]{Proposition}
\newtheorem{corollary}[theorem]{Corollary}

\newtheorem{definition}[theorem]{Definition}

\newdefinition{rmk}{Remark}
\newdefinition{assumption}{Assumption}
\newproof{proof}{Proof}
\newproof{pot}{Proof of Theorem \ref{thm2}}

\newcommand{\tDelta}{{\tilde\Delta}}
\newcommand{\txi}{{\tilde\xi}}
\newcommand{\trho}{{\tilde\rho}}

\newcommand{\hxi}{{\hat\xi}}

\newcommand{\wOmega}{{\widehat\Omega}}

\newcommand{\cxi}{{\check\xi}}
\newcommand{\crho}{{\check\rho}}

\newcommand{\bxi}{{\breve\xi}}
\newcommand{\brho}{{\breve\rho}}

\newcommand{\PP}{{\mathbb P}}

\newcommand{\II}{{\mathds 1}}
\newcommand{\EE}{{\mathbb E}}

\begin{document}

\begin{frontmatter}

\title {\textbf{Foreign Exchange Markets with Last Look}
\tnoteref{t1}\\[1em]
\textit{Mathematics and Financial Economics, Forthcoming}\\
{\small Cartea, {\'A}., Jaimungal, S. \& Walton, J. Math Finan Econ (2018). https://doi.org/10.1007/s11579-018-0218-3}
}
\tnotetext[t1]{SJ would  like to thank NSERC and GRI for partially funding this work. The authors thank C. Alexander, R. Anderson, F. Barnes, E. Benos, J. Danielsson, A. Gerig, PK Jain, C.A. Lehalle,  R. Oomen, H. Ralston,  Y. Sharaiha,  C. Vega, J.P. Zigrand, and seminar participants at the University of Sussex,  Norges Bank Investment Management, U.S. Securities Exchange Commission (SEC), U.S. Commodity Futures Trading Commission (CFTC), Board of Governors of the Federal Reserve System Washington DC, 6th Annual Stevens Conference on High-Frequency Finance and Analytics,  Market Microstructure CFM-Imperial 2015, the Systemic Risk Centre (LSE), University of Oxford, and the Bank of England for  comments on an earlier version of this article.}

\author[author1]{\'Alvaro Cartea}
\ead{alvaro.cartea@maths.ox.ac.uk}
\author[author2]{Sebastian Jaimungal}
\ead{sebastian.jaimungal@utoronto.ca}
\author[author3]{Jamie Walton}
\ead{jamie.walton@ucl.ac.uk}

\address[author1] {Mathematical Institute, University of Oxford, Oxford, UK\\  Oxford-Man Institute of Quantitative Finance, Oxford, UK}
\address[author2] {Department of Statistical Sciences, University of Toronto, Toronto, Canada}
\address[author3] {Department of Mathematics, University College London, London UK}

\begin{abstract}
We examine the Foreign Exchange (FX) spot price spreads with and without Last Look on the transaction. We assume that brokers are risk-neutral and they quote spreads so that losses to latency arbitrageurs (LAs) are recovered from other traders in the FX market. These losses are reduced if the broker can reject, ex-post, loss-making trades by enforcing the Last Look option which is a feature of some trading venues in FX markets. For a given rejection threshold the risk-neutral broker quotes a spread to the market so that her expected profits are zero. When there is only one venue, we find that the Last Look option reduces quoted spreads. If there are two venues  we show that the market reaches an equilibrium where traders have no incentive to migrate. The equilibrium can be reached with both venues coexisting, or with only one venue surviving. Moreover, when one venue enforces Last Look and the other one does not, counterintuitively, it may be the case that the Last Look venue quotes larger spreads.
\end{abstract}

\end{frontmatter}

\section{Introduction}

The Foreign Exchange (FX) marketplace has some unique structures which have lead to specific solutions for both exchanges and market makers. Unlike equities, there are less than 100 actively traded currencies and many can be traded across multiple platforms simultaneously. As there is no central exchange framework in FX, many Electronic Crossing Networks (ECNs) exist to service trading of currencies. The most common G10 currencies may be available to trade in more than 20 ECNs with multiple liquidity providers. Additionally, most major banks offer access to trade currencies through their own platforms either using an application or over an application programming interface (API) as well as through many ECNs.

In high-frequency trading, liquidity providers making markets on multiple streams are exposed to many risks. The technology race to reduce latency between exchanges has created an opportunity to extract value through latency arbitrage. This can manifest as a fast market participant trading on prices shown by slower liquidity providers in a rapidly updating market, and is not necessarily malicious. However, when the market taker is intentionally trading with the last liquidity provider to update her prices, or on stale quotes, then it may become necessary for the liquidity provider to construct a form of protection to prevent the misuse of her liquidity. 

A second concern for market makers is that they frequently show larger liquidity than what they have available. They do this because large market makers display prices/liquidity on multiple ECNs in the fragmented FX marketplace and at the same time provide streaming prices  to traders through APIs. This can mean that there exist thousands of potential streams where they are exposed to some notional amount of liquidity. Instantaneously, this liquidity does not represent the prices they are prepared to show in the full amount. Typically, however, if one-sided liquidity starts to be accessed on multiple venues simultaneously, then the market maker updates prices to all streams to reflect the new value of liquidity -- and ideally to attract traders to take them out of the risk by crossing some part of the spread. The risk therefore lies on the ability of the market maker to update prices on all streams in a rapid manner and thus is also at risk of latency arbitrage.

Generally, larger size trades have a larger bid-offer spread to represent the additional cost in trading out of the risk. In order to reduce transaction costs some traders may choose to split up a large order into smaller standard size amounts and hit liquidity on multiple venues simultaneously. This reduces the cost for the trader, but  exposes liquidity providers to the risk that the market will run away from them as they try to exit this position. In FX this activity is sometimes referred to as `spamming the market'.\footnote{There seems to be no general term in the FX industry that refers to the activity when a trader takes liquidity (same currency pair) in different venues at the same time. Here we use the term spamming or spraying for this type of activity. } The trader may also be accessing the same underlying source of liquidity on multiple venues if the best price on the ECNs is offered by the same provider. This is clearly a problem for the market maker.\footnote{If the trader were to request a quote for the full amount, rather than the child orders, the broker would quote a wider spread than that quoted for smaller orders. Wider spreads for large size FX orders are equivalent to large orders in equity markets walking the limit order book. }

There are some measures that market makers and ECNs can take to limit the exposure to  latency arbitrage strategies and to market takers spamming the market. In FX, some ECNs allow liquidity providers `Last Look': after a trader has traded on a market maker's price then the `Last Look' is a fixed period of time in which the market maker has an option to reject the trade. Generally the trade is  rejected if in this fixed period of time the trade moves against the market maker beyond some threshold. The market maker is inferring that the trader may be taking advantage of the liquidity and is essentially withdrawing the price they made to market. Doing so can neutralize the effect of a latency arbitrage as well as providing protection against market spamming, at least over the interval of time that Last Look is active, typically measured in milliseconds. Market makers may also use Last Look trade rejections on price streams provided to traders, particularly for traders who trade at a higher frequency.

In over-the-counter transactions  FX brokers stream quotes to a wide range of clients. A key characteristic that differentiates clients is their ability to see quote updates,  react to market news, and trade on the most up-to-date public information. Having access to low latency technology is expensive. FX brokers who stream prices recognize that not all clients have the capability of seeing the most recent quote and may come to the market trying to execute a trade on a stale quote at a price which could be advantageous to either the client or the broker. Thus, it is not unusual for brokers to allow trades on stale quotes, despite having streamed a new quote, because she wishes to attract order flow which could convey information that she may use to update her quotes.

The broker cannot  discern amongst the different strategies employed by an individual trader, in particular whether the trade is taking advantage of latency. For example, institutional investors often employ many strategies, some of which may involve latency arbitrages. Thus, Last Look is a measure designed for a type of strategy, not for a particular type of trader. In this paper we classify trades as either a latency arbitrage or non-latency arbitrage.  We allocate the latency arbitrage trades as the activity emanating from latency arbitrageurs (LAs),  and the other trades as activity from  slow traders (STs). Clearly, trades from market participants who employ both types of strategies will sometimes be  classified as coming from LAs and others from  STs. This slight abuse of nomenclature helps to clarify the setup of the model and discussion of the Last Look option in the rest of the paper.


Last Look is a controversial topic in the FX marketplace with some ECNs actively advertising that they do not allow Last Look liquidity providers on their platforms. However it does   protect market makers from more aggressive behavior and ultimately, prices offered on Last Look platforms may have lower spreads than on non-Last Look markets.   This means that  market participants who are not latency arbitraging the market maker are not penalized in the prices they receive, but may still face rejection of some of their trades. For direct pricing streams, employing trade rejection over Last Look also allows market makers to offer more liquidity to traders than they could without such protection. The disadvantage for traders is that they no longer have guaranteed fills when they go to market and, more pertinently, the rejected trades generally are the ones that have gone in their favor, at least over the Last Look time interval.


FX market makers  are exposed to being picked-off if they do not update their quotes quickly. However, some FX brokers willingly allow trades on stale quotes (e.g. in over-the-counter and quote streaming set-ups), but this is not a free option available to liquidity takers.  FX brokers `charge' for the  option, to be hit/lifted on stale quotes, by rejecting trades   through the Last Look mechanism -- see \cite{copeland1983information} who discuss firm quotes as free options given to market takers.

Our paper and the contemporaneous work of \cite{OomenLastLook} are the first to examine FX spot price spreads with and without Last Look on the transaction, see also  and \cite{OomenAgr}.  We model latency arbitrage by allowing the market taker to trade on a stale quote, which in FX markets is a quote that is no longer valid either because  the liquidity provider has sent an updated quote, or because the market has moved since the liquidity provider made the price. We consider the value to the liquidity provider of having the option to reject a quote over the Last Look interval given that there is a target rejection threshold which affects all traders. 

We assume that market makers or brokers are risk-neutral and competition drives spreads so that expected profits from dealing in the FX market are zero.    Brokers cannot observe the type of trade they are facing, so rejection affects all traders: LAs, who only trade on stale quotes which produce an immediate risk-less positive profit,  and STs,  who  are not (latency) arbitraging the market.  The brokers reject trades that generate losses greater than a predetermined threshold. These losses are calculated ex-post using the price update after the trader executed his order. As expected, the right to cancel trades over the rejection window  caps brokers' losses,  so everything else equal,  quoted spreads decrease.

We show that in markets where there is price momentum, i.e. price revisions are positively correlated (such as what occurs when there is spamming in the market), the broker's rejection rule is more effective at singling out latency arbitrage trades. Thus, everything else equal, when there is momentum in prices, spreads are tighter. Conversely, when price revisions are negatively correlated, prices mean revert and it is more difficult for the broker to single out loss-leading trades whose counterparty are LAs, hence spreads widen.

Tighter spreads  have different effects on market participants. LAs have more opportunities to attempt an arbitrage (on stale quotes),  because spreads are tighter and therefore LAs can take advantage of smaller price movements,  but they also face higher rejection rates and overall they are worse off in markets with the Last Look option. On the other hand, the STs  benefit from lower spreads, but face rejection of their most profitable trades, so depending on market parameters, how STs account for the foregone profits of rejected trades, and other rejection costs, they will seek or avoid trading in venues with Last Look.

Is there an optimal spread? In a market where there is only one venue to trade, the risk-neutral brokers are indifferent  between making markets with or without the Last Look option because spreads are determined by the zero expected profit condition. On the other hand, when STs account for rejection costs, our results show that there is an optimal spread that minimizes the STs' costs of executing round-trip trades. In addition to the spread that STs pay when executing trades, the rejection costs include:  forgone profits; immediacy costs which are high if the ST requires immediate and guaranteed execution; the additional cost arising from returning to the market to execute the trade; and, arguably,  the potential exposure to front-running costs.

When there is more than one FX venue, traders migrate to those where they are better off: LAs migrate to venues where the expected profit of a round-trip trade is highest, and STs to those where the expected cost of a round-trip trade is lowest. Quoted spreads depend on a number of factors which are specific to each venue: rejection rule, and proportion of LAs. We show that there is an equilibrium region where there are no incentives to migrate and also examine cases in which the equilibrium region is a corner solution where only one  FX venue survives, i.e. one venue attracts all order flow from both types of traders as well as all market makers.

In particular we discuss the two-venue case where in one venue brokers employ the Last Look option, while the other venue does not allow market makers to enforce Last Look. We show that there are two distinctive regions (defined by pairs of numbers of LAs and STs trading in each venue), where traders have incentives to migrate and the equilibrium reached is either both venues coexist or only one survives. When the market's starting point is in the region where the venue \emph{with} Last Look starts off with a low proportion of LAs, then equilibrium is reached when \emph{all} traders exit the venue without Last Look, i.e. all order flow occurs in the venue that employs a rejection rule.

The other region is one where  the venue \emph{without} Last Look starts with a low proportion of LAs (so the venue with Last Look has a high proportion of LAs). In this case,  LAs find it optimal to migrate to the venue without Last Look. Thus the brokers in the venue without Last Look increase spreads to recover the losses to LAs, but this increases the STs' trading costs, so some of them migrate to the venue with Last Look, but do so at a rate lower  than that at which LAs flow into the venue without Last Look. Equilibrium is reached at a point where both venues coexist (apart from very extreme cases  where the starting point is one where most LAs are concentrated in one venue). Interestingly, when both venues coexist  the Last Look venue does not always quote the lowest spread.

When traders switch between venues they incur a fixed cost. In the over-the-counter  FX market, this fixed cost includes `reputational' costs to build a relationship with the market maker, and software set-up costs to connect to other exchanges and counterparties. We show that when migration costs are very low, the market settles to an equilibrium where only one venue survives and this outcome depends on the starting point, but in most cases all traders migrate to the venue which enforces Last Look.

Finally, the Last Look feature in FX markets is in the spotlight of regulators and financial authorities. This paper provides a framework to analyze the provision of liquidity and immediacy in a market where some venues enforce rejection of trades.   For example, in a recent consultation document, the Bank of England (joint with the HM Treasury and the Financial Conduct Authority) express the concern raised by some market participants who ``have argued that such practices may also incentivize market makers
to delay a decision for longer periods in order to observe market moves and reject unprofitable trades or even engage in
front-running of orders.", \cite{BoE}. This paper provides a framework to understand how FX venues with different rejection rules set spreads to the market, thus providing a price for immediacy in the market, and how market participants choose venues for their trades.

The remainder of this paper is organized as follows. In Section \ref{sec: optimal spread no LL} we present the model for the dynamics of exchange rates and show how a risk-neutral broker sets optimal spreads in a market consisting of LAs and STs. In Section \ref{sec: slow spreads with LL} we develop the model further to allow the broker to enforce the Last Look option to cancel trades ex-post and determine the optimal spread quoted in the market. In Section \ref{sec: optimal spread for ST} we model how STs impute costs to rejected trades  and compute the optimal spread (hence the rejection threshold) that minimizes the costs that STs are exposed to. In Section \ref{sec: equilibrium venues} we discuss how the market reaches equilibrium when there is more than one FX venue.  Finally, Section \ref{sec: conclusions} concludes and proofs are collected in the Appendix.

\section{Optimal Spreads without Last Look}\label{sec: optimal spread no LL}

We assume that brokers  are risk-neutral and operate in a competitive market, so that the expected profits of round-trip trades  is zero. In addition,  brokers do not incur any fees or other variable costs to operate in the market.  The midprice, i.e.  the exchange rate between two currencies, follows a stochastic process which is observed by all market participants. There are three time markers $i=0,1,2$,  the midprice is denoted by $P_i$,   $P^a_i$ denotes the ask, $P^b_i$ the bid. The spread is given by  $\Delta= P^a_i- P^b_i \geq 0 $ and is determined by the brokers' zero-expected profit condition.  Point $i=0$ corresponds to the initial time when the broker posts a quote, $i=1$ corresponds to the time when the broker updates the quote, and $i=2$ corresponds to the time at which the broker decides whether to accept or reject the trade if there is a Last Look option. All trades are of one unit.

Throughout this paper the spread arises from the brokers' need to break-even when trading with market participants who arbitrage stale quotes.\footnote{When the trader hits the liquidity provider's most up-to-date quote, but the market has moved,  may also be considered as a trade on a stale quote.} In general, the difference between the bid and ask is explained by the various risks that the market maker or broker faces when intermediating trades, e.g.  adverse selection and  inventory risk, see for instance \cite{glostenMilgrom}, \cite{GrossmanMiller1988}, \cite{FrankDeJong2009}.   Here, we focus on the effect that LAs have on spreads, and one could include these other effects, which would widen the spreads.

Innovations in the midprice are given by
\[
P_{i+1}-P_i =\sigma\,Z_{i+1}\,,
\]
where $\sigma$ is a positive constant, the price revisions $Z_1$ and $Z_2$ are correlated standard normal random variables, with correlation coefficient $\rho$, and we write,
\[
\begin{pmatrix}
Z_1 \\
Z_2
\end{pmatrix}
\sim
\mathcal N\left(
\begin{pmatrix}
0 \\
0
\end{pmatrix},
\begin{pmatrix}
1 & \rho \\
\rho & 1
\end{pmatrix}\right)\,.
\]
Positive correlation, $\rho>0$,  corresponds to a period of trading where prices are trending up/down, while negative correlation,  $\rho<0$,  corresponds to a time of mean-reversion of prices. Naturally, there is no trend in prices when correlation is zero. In this section the broker does not have the Last Look option to veto trades ex-post, so the second price increment is irrelevant, it will however play an important role when this option is incorporated in Section \ref{sec: slow spreads with LL}.

When there is spamming in the market, i.e. when an LA takes liquidity from multiple venues simultaneously, price updates reflect this type of market activity by moving in the direction of the trade. Consider the case of an LA submitting buy orders over multiple venues  (and possibly from different brokers) simultaneously. Several brokers will then be left with excessive short positions that they must unwind. To do so, the brokers will either take liquidity and thus add to overall buying pressure in the market resulting in upward price movements; and/or adjust their bids (and hence also asks) upwards to entice other traders to offset their short position. The end result is that prices  move upwards and this pressure can persist over multiple periods depending on the size of the total short position the brokers found themselves in. A similar argument follows if the LA submits sell orders over multiple venues simultaneously, resulting in a downward trend in prices. Overall, spamming in the market induces positive correlation between price increments.

All brokers  send   quote updates  at the beginning of every period $i$ and traders decide if they want to trade. The market is populated by two types of traders: STs and LAs. STs do not possess the technology to always observe the updates that the brokers post. LAs have the speed and technology to see, and act on, all quote updates to the market.

The brokers cannot differentiate trader type, but know that a proportion $\alpha\in[0,1]$ of traders are LAs, and know that STs observe the updated quote (at $i=1$) with probability $\beta$. The brokers wish to do business with STs, so they allow all market participants to trade on stale quotes. This may happen in two ways.  i) At time $i=1$ a broker updates her quotes to $P^a_1=P_1+\frac{\Delta}{2}$ and $P^b_1=P_1-\frac{\Delta}{2}$, but will honor trades at the stale quotes $P^{a,b}_0$. ii) At time $i=1$ the market has moved and a broker did not update her quotes and will honor trades at the stale quotes $P^{a,b}_0$. In the sequel, a trade on a stale quote refers to either one of these cases.  Throughout we refer to $\alpha$ as the proportion of traders, but could also be interpreted as the ratio of latency arbitrage trades to the total number of trades in the FX market.

An ST always trades at the quotes he sees, whether  stale or not.  LAs will always trade at the most favorable quote for him, stale or new. Thus, brokers are exposed to `latency losses' when trading with LAs who take advantage of stale quotes. In equilibrium, brokers set the spread $\Delta$ to recover these losses.

\subsection{Optimal spread}

The broker determines the quoted  spread so that the expected profit of each round-trip trade, in any given period, is zero. When the broker enters a position at time $i=1$ the expected profit of the round-trip is calculated using the price at which the first leg of the trade is entered, and the   price of the leg to close out the position. The former depends on whether the broker accepted the trade on a stale or updated quote. The latter is either $P^b_1$, if first leg was a sell, or $P^a_1$, if the first leg was a buy.

Figure \ref{fig:three cases for LA} shows quote updates. The size of the spread  and the midprice change determine if the LA trades on a stale quote. Cases I and II show arbitrage opportunities executed by LAs. Panel (c) depicts the cases where the midprice change is small enough to preclude latency losses to the broker.

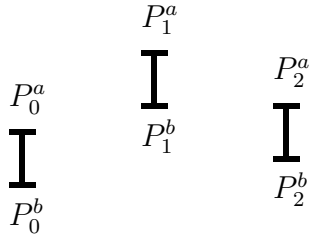
\begin{figure}[t!]

\begin{center}

\begin{minipage}{0.45\textwidth}
\linethickness{2pt}
\begin{picture}(100,120)(0,-40)
\put(0, 50){$P_0^a$}
\put(5, 40){\line(0,-1){20}}
\put(0, 40){\line(1,0){10}}
\put(0, 20){\line(1,0){10}}
\put(0, 5){$P_0^b$}

\put(50, 80){$P_1^a$}
\put(55, 70){\line(0,-1){20}}
\put(50, 70){\line(1,0){10}}
\put(50, 50){\line(1,0){10}}
\put(50, 35){$P_1^b$}

\put(100, 60){$P_2^a$}
\put(105, 50){\line(0,-1){20}}
\put(100, 50){\line(1,0){10}}
\put(100, 30){\line(1,0){10}}
\put(100, 15){$P_2^b$}
\end{picture}
\\
{(a) Case I: $P_1^b>P_0^a$}
\end{minipage}
\begin{minipage}{0.45\textwidth}
\caption{A sequence of bid-ask price updates. The first quote is at $i=0$, the updated quote at $i=1$, and the third update at $i=2$ is used to determine the Last Look rejection. \label{fig:three cases for LA} \vspace{8em}}
\end{minipage}
\\
\begin{minipage}{0.45\textwidth}
\linethickness{2pt}
\begin{picture}(100,120)(0,-40)
\put(0, 50){$P_0^a$}
\put(5, 40){\line(0,-1){20}}
\put(0, 40){\line(1,0){10}}
\put(0, 20){\line(1,0){10}}
\put(0, 5){$P_0^b$}

\put(50, 20){$P_1^a$}
\put(55, 10){\line(0,-1){20}}
\put(50, 10){\line(1,0){10}}
\put(50, -10){\line(1,0){10}}
\put(50, -25){$P_1^b$}

\put(100, 60){$P_2^a$}
\put(105, 50){\line(0,-1){20}}
\put(100, 50){\line(1,0){10}}
\put(100, 30){\line(1,0){10}}
\put(100, 15){$P_2^b$}
\end{picture}
\\
(b) Case II: $P_1^a<P_0^b$
\end{minipage}
\begin{minipage}{0.45\textwidth}
\linethickness{2pt}
\begin{picture}(100,120)(0,-40)
\put(0, 50){$P_0^a$}
\put(5, 40){\line(0,-1){20}}
\put(0, 40){\line(1,0){10}}
\put(0, 20){\line(1,0){10}}
\put(0, 5){$P_0^b$}

\put(50, 55){$P_1^a$}
\put(55, 45){\line(0,-1){20}}
\put(50, 45){\line(1,0){10}}
\put(50, 25){\line(1,0){10}}
\put(50, 10){$P_1^b$}

\put(100, 60){$P_2^a$}
\put(105, 50){\line(0,-1){20}}
\put(100, 50){\line(1,0){10}}
\put(100, 30){\line(1,0){10}}
\put(100, 15){$P_2^b$}
\end{picture}
\\
(c) Case III:
$P_1-P_0 \in [- \Delta, \Delta]$
\end{minipage}
\end{center}
\end{figure}

To determine the broker's optimal spread we first look at the trades where the counterparty is an ST and then when it is an LA.

\textbf{Trading with STs}

Recall that the  ST sees the updated quote at $t=1$ with probability $\beta$.
\begin{itemize}
\item If  the ST receives the updated quote, then the  profit to the broker of a round-trip trade is the spread $\Delta$.
\item If the ST does not receive the updated quote, and therefore trades on the stale quote, the profit to the broker of a round-trip trade is
    \[
     P_1-P_0+\Delta \,.
    \]
    Clearly, when the  ST trades on a stale quote it will be, unbeknownst to him,  at a profit or at a loss.
\end{itemize}

\textbf{Trading with LAs}

Trades on stale quotes result from options provided by the broker to liquidity takers who exercise them. In equity markets, firm quotes in the limit order book are `free' options given to liquidity takers to pick-off stale quotes. In FX markets  with Last Look these options are not free because the broker  may reject trades.

Here we list the midprice revisions which expose the broker to latency losses:
 \begin{itemize}
 \item \underline{Case I}: If $P_1^b>P_0^a$, the LA executes a buy at the stale quote, followed by (an instant later) a sell at the updated quote, and  the LA receives a net profit of
     \[
     \left( P_1^b-P_0^a\right)_+\,,
     \]
    where $(x)_+=\max(0,\, x)$.
 \item \underline{Case II}: If $P_1^a<P_0^b$, the LA executes a sell at the stale quote, followed by (an instant later) a buy at the updated quote, and  the LA receives a net profit of
     \[
     \left( P_0^b-P_1^a\right)_+\,.
     \]

\end{itemize}
And midprice revisions which do not lead to latency losses:
 \begin{itemize}
 \item \underline{Case III}: If $P_1-P_0 \in [- \Delta, \Delta]$, the LA cannot profit from a round-trip trade and therefore makes no trades.
\end{itemize}

Putting the above scenarios together, the broker's expected profits  stemming from trading with  STs and LAs, respectively, are:
\begin{equation}\label{eqn: profits from STs no LL}
\Omega_{ST}= \beta \,\Delta+ (1-\beta)\,\EE_0[P_1-P_0+\Delta]\,,
\end{equation}
and
\begin{equation}\label{eqn: profits from LAs no LL}
\Omega_{LA}=\EE_0\left[\left( P_1^b-P_0^a\right)_+ \,
+\,
\left( P_0^b-P_1^a\right)_+\right]\,,
\end{equation}
where $\EE_0$ is the expectation operator conditioned on information at time $i=0$.

Thus, the broker's expected profits at time $i=0$ are given by
\begin{equation}\label{eqn: total profits no LL}
\Omega = (1-\alpha)\,\Omega_{ST}\,-\,
\alpha\,\Omega_{LA}\,.
\end{equation}

Next, we determine the balancing equation that the spread must satisfy. Recall the broker is risk-neutral and does not incur any fees or other variable costs to make markets. Thus, in equilibrium, the broker  sets a spread where the expected profit is zero. We seek the optimal spread by  conditioning on type of trader.

First, due to the martingale nature of the price movement  over the first period,  the expected profit from trading with  STs is
\[
\Omega_{ST} = \beta \,\Delta+ (1-\beta)\,\EE_0[P_1-P_0+\Delta] = \Delta\,.
\]

Second, we can rewrite the expected profits from trading with  LAs as follows:
\begin{align*}
\Omega_{LA}
=&\, \EE_0\left[\left( P_1^b-P_0^a\right)_+ + \left( P_0^b-P_1^a\right)_+\right] \\
=&\, \EE_0\Big[\left( P_1-P_0-\Delta\right)_+ + \left( P_0-P_1-\Delta\right)_+\Big] \\
=&\, 2\,\EE_0\Big[\left( P_1-(P_0+ \Delta)\right)_+ \Big]\,.
\end{align*}
In this form, we can interpret the expected profits from trading with  LAs  as two call options on the midprice struck at the arrival price plus the spread, or alternatively as a single strangle option at the same strike. Since we assume prices are arithmetic, and increments are symmetric, these two options have the same value.

\begin{proposition}\textbf{Losses to Latency Arbitrageurs without Last Look.}
\label{thm: loss to LA without lastlook}
 The broker's expected losses to LAs  are given by
%
\begin{equation}
\Omega_{LA} = 2\,\sigma\,\phi\left(\frac{\Delta}{\sigma}\right)
-2\,\Delta\,\Phi\left(-\frac{\Delta}{\sigma}\right) \,,
\end{equation}
where $\phi(\cdot)$ and $\Phi(\cdot)$  denote the standard normal pdf and cdf, respectively.
\end{proposition}
\begin{proof}
See \ref{sec: proof of loss to LA without lastlook}.
\end{proof}

In equilibrium, the broker must break-even so the losses she incurs from trading with LAs must be offset by the gains obtained from trading with STs. Thus, the broker must quote a spread to the market so that $\Omega=0$, so using  \eqref{eqn: total profits no LL}, the zero-expected profit condition is $\alpha \,\Omega_{LA} = (1-\alpha)\,\Omega_{ST}$. This is shown in the following corollary.

\begin{corollary}\textbf{Optimal Spread Balancing Equation without Last Look.} The risk-neutral broker charges a spread $\Delta^*=\sigma\,x^*$, where $x^*$ is a solution of the non-linear equation
\begin{equation}
\phi\left(x\right)-x\,\Phi\left(-x\right) =\frac{1-\alpha}{2\,\alpha}\,x\,.
\label{eqn:balance no lastlook}
\end{equation}
\end{corollary}
\begin{proof}
Setting the broker's expected profits to zero $\Omega=(1-\alpha)\,\Omega_{ST} - \alpha\,\Omega_{LA}=0$, and rearranging, leads directly to the above balancing equation. \qed
\end{proof}
Moreover, the proposition below shows that there is a unique optimal spread where \eqref{eqn:balance no lastlook} holds.
\begin{proposition}
\label{prop: unique finite solution for spread no last look}
There exists a unique finite solution $x\in[0,+\infty)$ to the non-linear equation \eqref{eqn:balance no lastlook} if and only if $\alpha\in[0,1)$.
\end{proposition}
\begin{proof}
See \ref{sec: proof of unique finite solution for spread no last look}.
\end{proof}

It is clear that STs bear the costs imposed on the market by the LAs who trade on stale quotes. Figure \ref{fig:Delta Opt No Lastlook} shows a plot of the optimal spread $\Delta^*$ as a function of the percentage $\alpha$ of LAs in the market. As expected, this optimal spread is increasing in $\alpha$.  The diagram stops at  $\Delta^*=2\,\sigma$, however, there is indeed a vertical asymptote at $\alpha=1$; it is simple to see that as $\alpha \to 1$,  the solution  of \eqref{eqn:balance no lastlook} is $x^*\to\infty$.

\begin{SCfigure}[20][!t]
\hspace{5em}\includegraphics[width=0.45\textwidth]{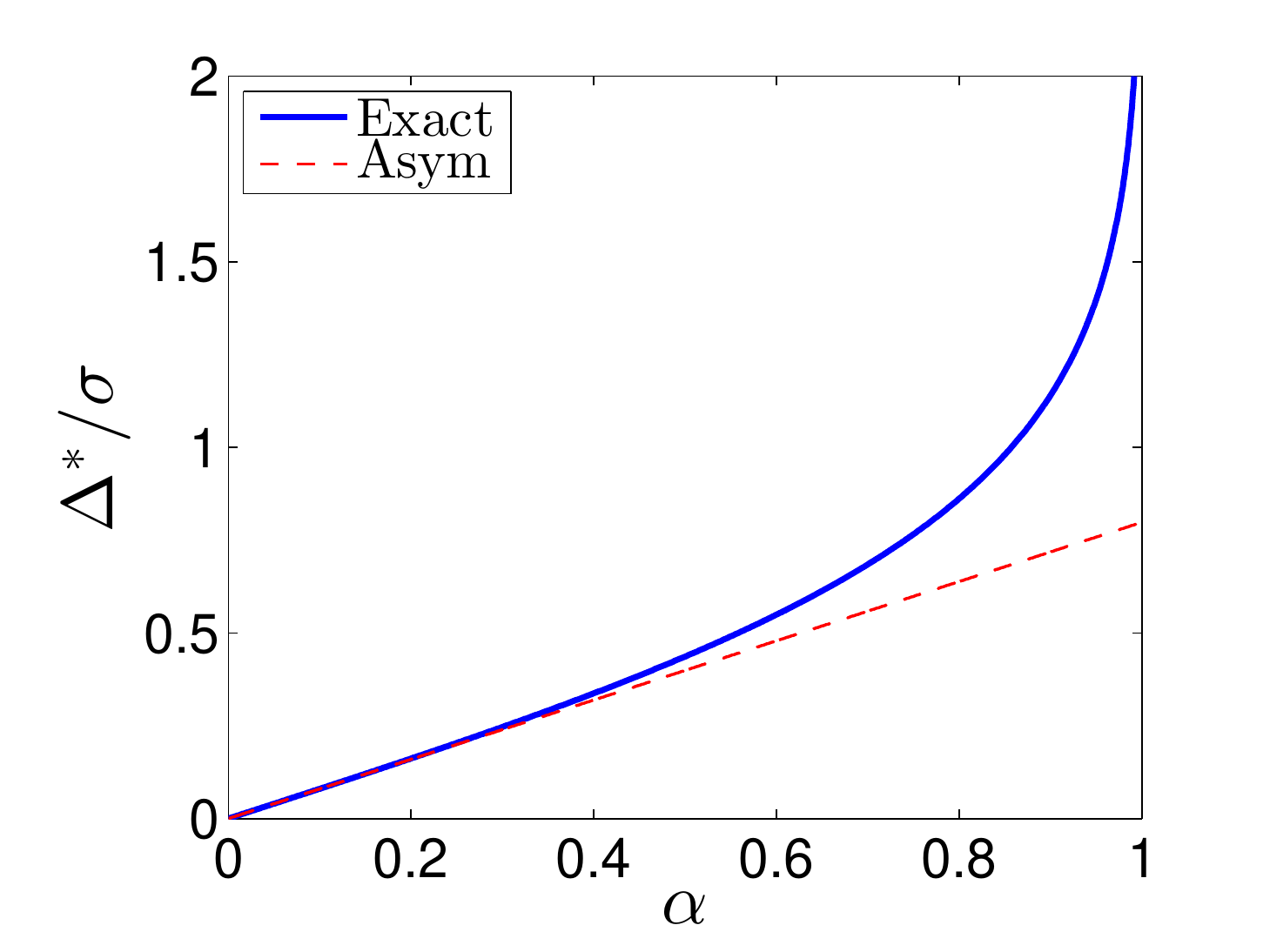}
\caption{The optimal spread $\Delta^*$ (relative to $\sigma$) which renders the broker's expected losses to LAs equal to her expected gains from STs. Recall that $\alpha$ is the percentage of LAs in the market. \vspace{4em} \label{fig:Delta Opt No Lastlook}}
\end{SCfigure}

\begin{proposition}\textbf{Asymptotic Optimal Spread.}
\label{prop: asymp slow spread no last look}
When the proportion of LAs trading in the market is small, i.e.  $\alpha$ is small, the asymptotic solution of the optimal spread is
\begin{equation}
\frac{\Delta^*}{\sigma}= \sqrt{\frac{2}{\pi}}\;\alpha + o(\alpha)\,,
\label{eqn: spread no last look asymp}
\end{equation}
to first order.
\end{proposition}
\begin{proof}
See \ref{sec: proof of asymp slow spread no last look}.
\end{proof}
The dashed line in  Figure \ref{fig:Delta Opt No Lastlook} shows the asymptotic solution. This asymptotic form  has a connection to the \cite{glostenMilgrom} (GM) model. To see this, note that $\EE\left[|Z|\right]=\sqrt{\frac{2}{\pi}}$, where $Z$ is a standard normal random variable, so that if we identify $\sqrt{\frac{2}{\pi}}\;\sigma\sim \left(\overline{V}-\underline{V}\right)$ where $\overline{V}$, $\underline{V}$ are the two possible price outcomes in the GM model, then from \eqref{eqn: spread no last look asymp}, we have $\Delta^* \sim \alpha \,\left(\overline{V}-\underline{V}\right)$. This result corresponds to the spread in the GM approach when $\alpha$ represents the percentage of informed traders in the market.

\section{Optimal Spread with Last Look}\label{sec: slow spreads with LL}

In this section we employ the same framework as the one developed above. As before,  brokers allow market participants to trade on stale quotes, but brokers have the option of cancelling trades ex-post. Recall that brokers do not know the type of trader they are doing business with, so trades are rejected when the losses to the broker exceed a predetermined threshold which is the same for all brokers. The sequence of events is as follows.

LAs will only trade if midprice updates are such that they can make an immediate risk-less profit (Cases I and II in Figure \ref{fig:three cases for LA}), which requires the first trade of their latency arbitrage to be on the stale quote -- the second leg of their arbitrage is at the current quote $P^{a,b}_1$. STs on the other hand, trade on stale quotes only when they did not receive the updated quote.  In either case, let $P_e$ denote the midprice at which  the trader executed his first trade. Then the broker employs the following ex-post rejection rule at time $i=2$. If the trader sells to the broker, the broker rejects the trade if $-P_e+P_2 \le \xi$ (with the threshold $\xi<0$), while if the trader  buys from the broker, the broker rejects the trade if $P_e - P_2 \le \xi$, i.e. the broker rejects trades when her losses are larger than the threshold $|\xi|$ net of the spread cost that they pick up.\footnote{When the broker receives a buy order, she sells the asset so her cash increases by $P_e$ plus the half-spread, and at period $i=2$ she uses the midprice $P_2$ minus the half-spread to decide if the trade is rejected. Thus, the broker rejects the trade if her losses to this round-trip trade are less than $\xi-\Delta$. So if the trader buys a share on the quote  at $i=1$, then the broker rejects it if $(P_e+\Delta/2) -(P_2-\Delta/2) \le \xi +\Delta$.}

Here we assume that there is only one venue and  the rejection threshold is set by the venue. The choice of threshold  does not affect the brokers' business because, conditioned on the threshold $\xi$, brokers set spreads to break even. In addition, the choice of threshold does  not alter the fraction of LAs and STs that the brokers face because there is only one venue to trade. Later,   in Section \ref{sec: equilibrium venues} we examine in detail what happens when there is more than one venue.

In the following subsection we discuss the ST's costs of round-trip trades conditioned on the fact that they were accepted, and  in Section \ref{sec: optimal spread for ST} we discuss how STs calculate costs of round-trip trades by also imputing a cost to rejected trades.

\subsection{The Slow Trader's Cost}

If the ST receives the updated quote (with probability $\beta$), then a round-trip trade costs him the spread $\Delta$. If he buys (which we assume occurs $50\%$ of the time), his trade will only be \emph{accepted} if $P_e-P_2=P_1-P_2>\xi$. Similarly, if he sells, his trade will only be \emph{accepted} if $P_2-P_e=P_2-P_1>\xi$. In all, the ST's expected cost of a round-trip trade when he receives the updated quote is
\begin{equation}
\Omega_{ST\,|\,\text{updated}}
=\tfrac{1}{2} \, \Delta \, \PP\left[P_1-P_2>\xi\right]
+ \tfrac{1}{2} \, \Delta \, \PP\left[P_2-P_1>\xi\right]
=\Delta\;\Phi\left(-\frac{\xi}{\sigma}\right)\,.
\label{eqn: Omega ST on update}
\end{equation}

If the ST does not receive the updated quote, then a round-trip trade costs him $\left(P_0+\frac{\Delta}{2}\right) - \left(P_1-\frac{\Delta}{2}\right)$ if he buys (then sells), and his trade is accepted only if $P_e-P_2=P_0-P_2>\xi$. Similarly for the case when the trader sells (then buys). In all, the ST's expected cost, given that he does not receive the updated quote, is
\begin{align}
\notag \Omega_{ST\,|\,\text{stale}}
=&\, \tfrac{1}{2} \, \EE\left[\left(P_0-P_1+\Delta\right)\,\II_{\{P_0-P_2>\xi\}}\right]
 + \tfrac{1}{2} \, \EE\left[\left(P_1-P_0+\Delta\right)\,\II_{\{P_2-P_0>\xi\}}\right]
 \\
=&\, \sigma\,\sqrt{\frac{1+\rho}{2}}\;\phi\left(\frac{1}{\sqrt{2(1+\rho)}}\,\frac{\xi}{\sigma}\right)
+\Delta\;\Phi\left(-\frac{1}{\sqrt{2(1+\rho)}}\,\frac{\xi}{\sigma}\right)\,.
\label{eqn: Omega ST on stale}
\end{align}
See  \ref{sec: derive Omega ST on stale} for the detailed computation.

\begin{proposition}\textbf{Cost to Slow Traders with Last Look.}
\label{thm: cost to ST with last look} The cost of a round-trip trade  by an ST when the broker has the Last Look option is
\begin{equation}
\begin{split}
\Omega_{ST}=&\, \sigma\,(1-\beta)\,\sqrt{\frac{1+\rho}{2}}\;\phi\left(\frac{1}{\sqrt{2(1+\rho)}}\,\frac{\xi}{\sigma}\right) \\
&\,+\Delta\left\{\beta\,\Phi\left(-\frac{\xi}{\sigma}\right)+(1-\beta)\,\Phi\left(-\frac{1}{\sqrt{2(1+\rho)}}\,\frac{\xi}{\sigma}\right)\right\}\,.
\end{split}
\label{eqn: Omega ST exact}
\end{equation}
\end{proposition}
\begin{proof}
This follows immediately from \eqref{eqn: Omega ST on update} and \eqref{eqn: Omega ST on stale}. \qed
\end{proof}

\begin{proposition}\textbf{Probability of a Slow Trader's Execution.}
\label{prop: Prob ST exec} The probability that the ST's trade is executed equals
\begin{equation}\label{eqn: Psi ST}
\Psi_{ST} = \PP[ P_e - P_2 > \xi] = \beta\,\Phi\left(-\frac{\xi}{\sigma}\right)
+(1-\beta)\,\Phi\left(-\frac{\xi}{\sigma\,\sqrt{2(1+\rho)}}\right)\,.
\end{equation}
\end{proposition}
\begin{proof}
See \ref{sec: Proof of Prob ST exec}.
\end{proof}
This probability is independent of the quoted spread because STs are not attempting to latency arbitrage the broker by trading on stale quotes.

\subsection{The Latency Arbitrageur's Profit}

The LA uses the same strategy as he did without the Last Look clause. He only trades if, relative to the stale quote, he can make a risk-less and profitable round-trip trade. Thus, whenever the LA executes a trade he  always does the first leg at the bid or ask posted in the previous period, i.e.  $P^{a,b}_0$. However, since the broker rejects trades, the   LA's expected profit  of a round-trip trade is
\begin{equation}\label{eqn: profits from LAs with LL}
\Omega_{LA} = 2\,\EE_0\left[ \; (P_1-P_0-\Delta)_+\,\II_{\{\,P_0-P_2> \xi\, \}}\; \right]\,,
\end{equation}
which is as \eqref{eqn: profits from LAs no LL}, but including  the indicator function $\II_{\{\,P_0-P_2> \xi \, \}}$ to account only for accepted trades.

\begin{proposition}\textbf{Losses to Latency Arbitrageurs with Last Look.}
\label{thm: Losses to LA with last look}
The expected losses that the broker, who employs the Last Look option, incurs when trading with LAs is
\begin{equation}
\Omega_{LA} = 2\,(B(\tDelta)- A(\tDelta)\,\tDelta)\,\sigma\,,
\label{eqn: Omega LA exact}
\end{equation}
where $\tDelta = \frac{\Delta}{\sigma}$,  $\tilde{\xi}=\frac{\xi}{\sigma}$,
\begin{align}
\nonumber A(\tDelta)&:=\PP[\;P_1-P_0>\Delta\,,\,P_0-P_2> \xi\;] \\
&=\, \Phi\left(-\frac{\txi}{\sqrt{2(1+\rho)}}\right)
-\Phi_{\sqrt{\frac{1+\rho}{2}}}\left(\tDelta\;,\;-\frac{\txi}{\sqrt{2(1+\rho)}}\right)\,,
\label{eqn:LA Prob Exec I}
\end{align}
and
\begin{align}
\nonumber B(\tDelta)&:= \EE_0\left[ \; \tfrac{1}{\sigma}(P_1-P_0)\,\II_{\{\,P_1-P_0>\Delta\,,\,P_0-P_2> \xi\,\}}\; \right] \\
&= \,
\phi(\tDelta)\,
\Phi\left(-\frac{\txi+(1+\rho)\tDelta}{\sqrt{1-\rho^2}}\right)
-\sqrt{\frac{1+\rho}{2}}\,\phi\left(\frac{\txi}{\sqrt{2(1+\rho)}}\right)
\Phi\left( -\frac{\txi+2\,\tDelta}{\sqrt{2(1-\rho)}} \right)\,.
\label{eqn:LA Prob Exec II}
\end{align}

\end{proposition}
\begin{proof}
See \ref{sec: Proof of Losses to LA with last look}.
\end{proof}

\subsection{Optimal spread with Last Look}

\begin{figure}[!t]
\begin{center}
\includegraphics[width=0.45\textwidth]{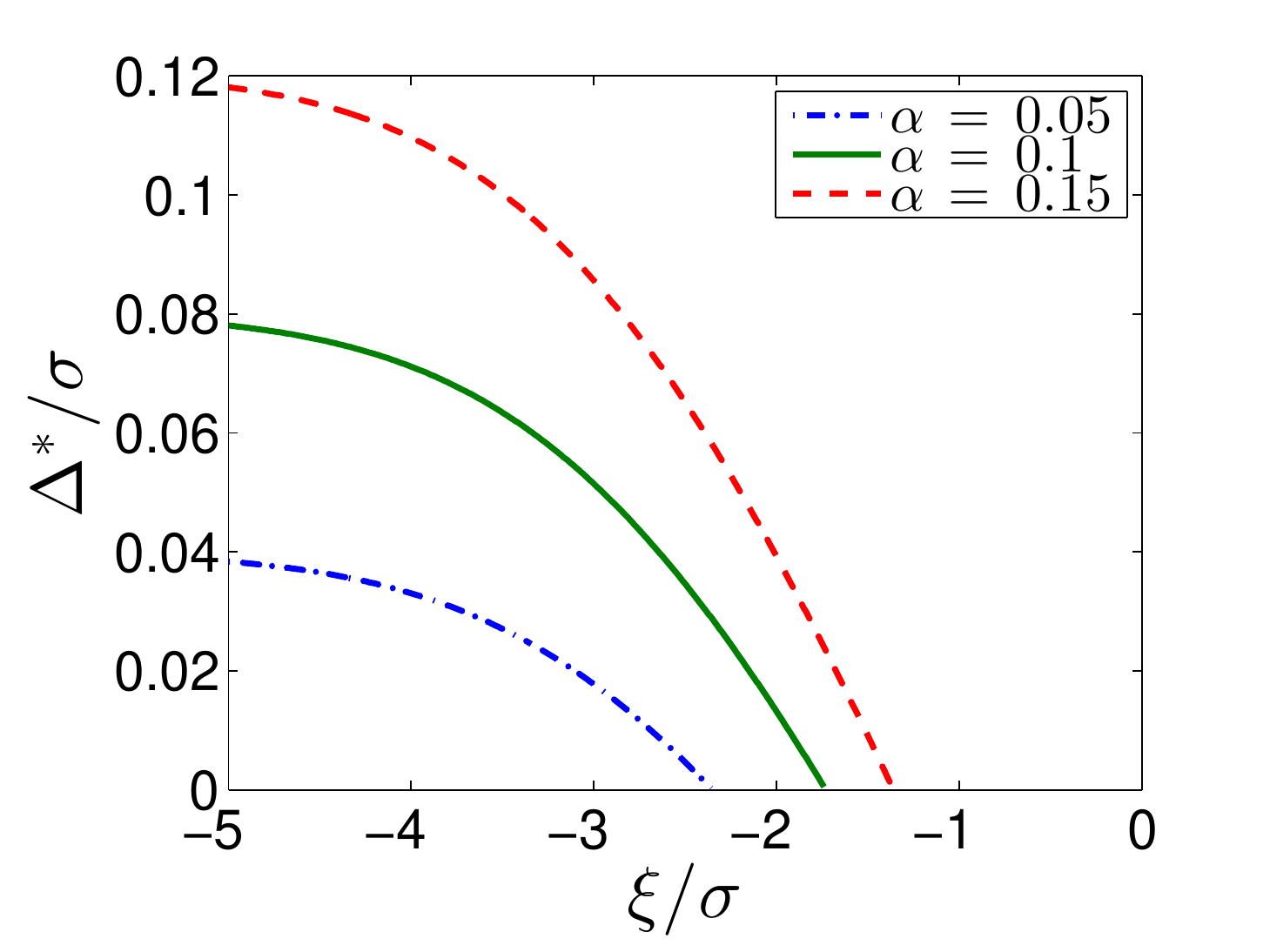}
\includegraphics[width=0.45\textwidth]{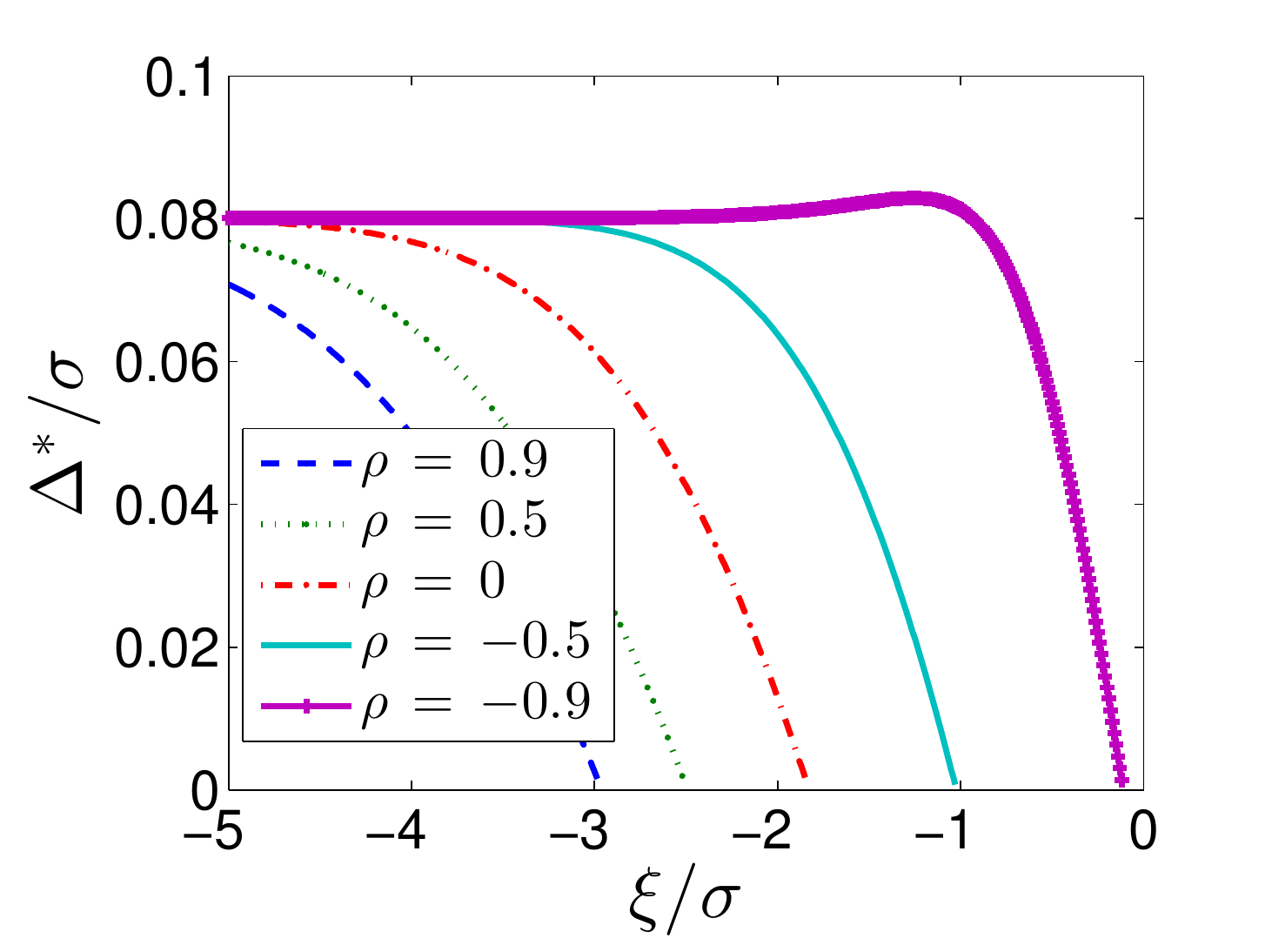}
\end{center}
\caption{Optimal spread $\Delta^*$ (relative to $\sigma$) which renders the broker's expected loss to LAs equal to her expected gains from STs. Recall that $\alpha$ is the percentage of LAs in the market. Here, $\beta=0.8$, in the left panel $\rho=0.5$, and in the right panel $\alpha=0.1$. \label{fig:Delta Opt}}
\end{figure}
Figure \ref{fig:Delta Opt} shows  the optimal spread  as a function of the rejection threshold $\xi$. Recall that the optimal spread is set such that the broker has zero expected profit and satisfies
\begin{equation}
(1-\alpha)\,\Omega_{ST}(\Delta) - \alpha\,\Omega_{LA}(\Delta) =0\,,
\label{eqn: optimal spread balancing equation}
\end{equation}
and all brokers use the same threshold $\xi$, which is determined by the venue.

The left panel shows how the optimal spread (normalized by the volatility parameter $\sigma$) depends on the  percentage $\alpha$  of LAs  trading in the market (correlation is fixed at $\rho=0.5$) and the rejection threshold imposed by the venue. The right panel shows how the optimal spread depends on the  correlation between the shocks to the midprice (percentage of LAs is fixed at $\alpha=0.1$). In both panels the optimal  spread decreases as the cutoff $\xi$ increases. This result reflects the fact that LAs make less profits from the broker because as $\xi$ increases, more trades are rejected -- the broker transfers less losses to the STs by charging a  smaller spread to the market. Furthermore, it is clear that the optimal spread is bounded above  (this bound is obtained when $\xi\to-\infty$) by the optimal spread in the absence of the Last Look option.

The figure also shows that there is a critical cutoff level $\xi^*$ which renders the optimal spread equal to zero, and  as the percentage of LAs increases, the optimal  spread increases -- this is natural, as the broker must recover the costs that the additional LAs impose on her. With the Last Look option,  brokers can remove the cost to STs entirely (i.e. spread is set at zero) because they are able to recover those costs by rejecting trades from the LAs. Note however, that with the Last Look option the costs  of only accepted trades from STs is reduced to zero, but the most  profitable trades executed by the ST are cancelled -- we return to this point in Section \ref{sec: optimal spread for ST} where the ST internalizes the costs of rejected trades.

Finally, we observe that when there are trends or momentum in the market, the Last Look feature singles out a higher proportion of LAs' trades.  For example, as correlation between midprice revisions increases,  when an LA profits in the first increment, this profit will also be reflected in the increment over the second period, which is when brokers enforce the ex-post rejection option,  and hence the rejection rule will pick them out better. The same argument shows that when correlation is negative, prices mean revert, it is more difficult for brokers to use the ex-post price to decide when to reject loss-leading trades executed by LAs, so spreads for a fixed rejection threshold are wider.

Next, we investigate how effective is the Last Look option at rejecting trades from LAs and not those stemming from STs. For this, we need the two results in the following propositions.
\begin{proposition}\textbf{Probability of a Latency Arbitrageur's Execution.} The probability that the LA's trade is executed is
\[
\Psi_{LA} = \PP\Big[ \,(P_0 - P_2) > \xi\,\Big|\, (P_1-P_0)>\Delta\,\Big] =
\frac{A}{\Phi\left(\frac{\Delta}{\sigma}\right)}\,,
\]
where $A(\cdot)$ is given in \eqref{eqn:LA Prob Exec I}.
\end{proposition}
\begin{proof}
Due to symmetry, we need only look at the case when the sell is at the stale and buy at the updated quote. The result above then follows immediately from the definition of conditional probabilities and using the result in \eqref{eqn:LA Prob Exec I}. \qed
\end{proof}
\begin{proposition}\label{prop:prob LA given reject}\textbf{Rejecting Latency Arbitrageur's Execution.} The probability that a trader was an LA given that the trade was rejected is
\[
\Upsilon = \PP[\text{LA}\,|\, \text{reject}\,] = \alpha \;\frac{1-\Psi_{LA} }{1-(\alpha\,\Psi_{LA}+(1-\alpha)\,\Psi_{ST})}\,.
\]
\end{proposition}
\begin{proof}
A straightforward application of Bayes' Theorem implies that
\begin{align}
\PP[\text{LA}\,|\, \text{trade rejected}\,]
=&\,  \alpha \;\frac{\PP[\text{reject} \,|\, \text{trade LA}\,]}{\PP[\text{reject}\,|\,\text{trade}]}\,,
\end{align}
and the result follows. \qed
\end{proof}
\begin{SCfigure}[20][!t]
\hspace{5em}
\includegraphics[width=0.45\textwidth]{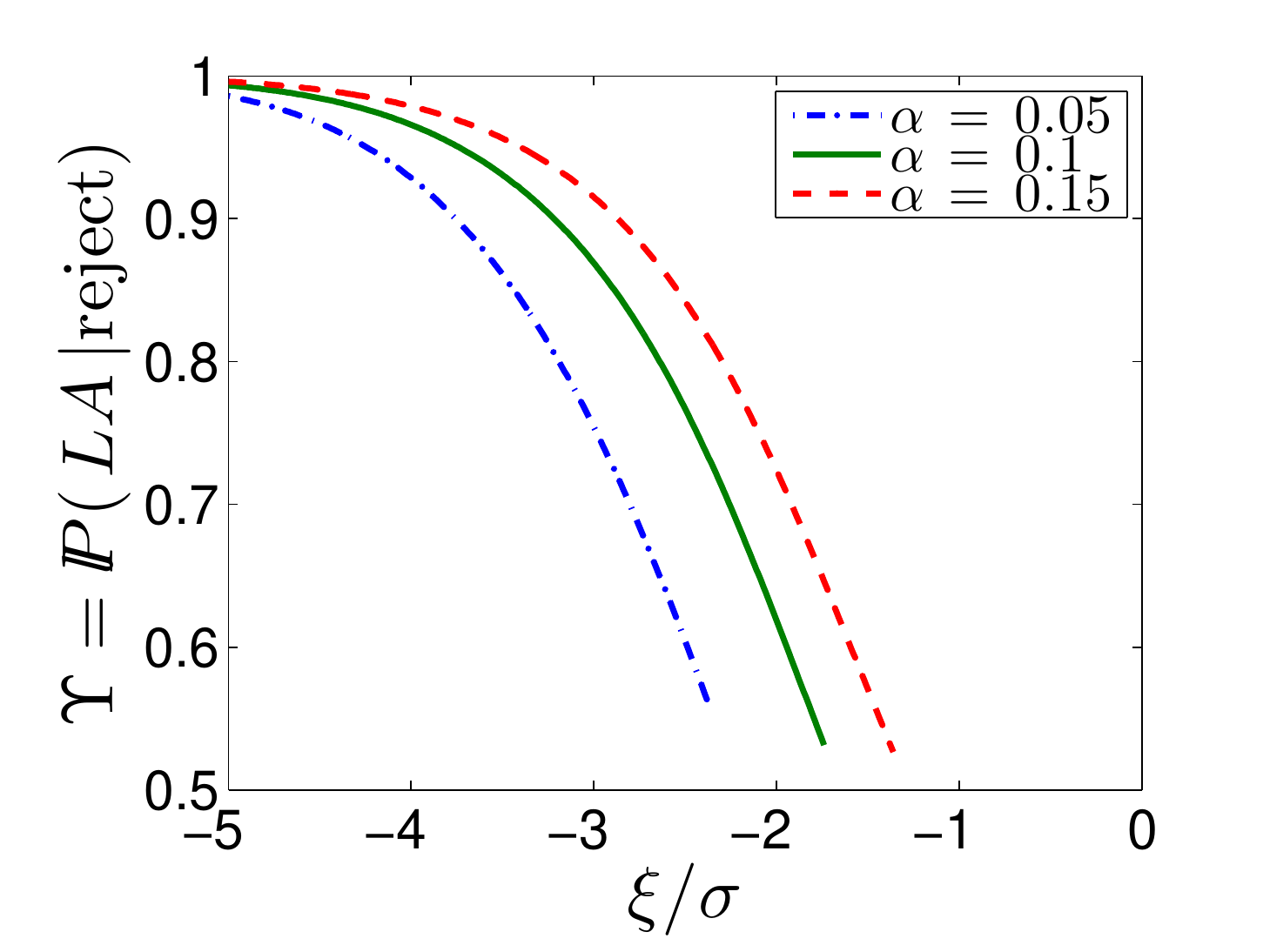}
\hspace{2em}
\caption{The probability that a trader was an LA given that the trade was rejected. \vspace{8em}} \label{fig:Upsilon}
\end{SCfigure}
In Figure \ref{fig:Upsilon}, we plot the probability that the agent was an LA, given that the trade was rejected, as a function of the cutoff $\xi$. For each level of $\xi$, we first determine the optimal spread as in Figure \ref{fig:Delta Opt}, and then compute $\Upsilon$ from Proposition \ref{prop:prob LA given reject}. The plot shows this is a decreasing function of $\xi$, and can be interpreted as follows: as the rejection threshold $\xi$ increases, so that more trades are rejected, it is more difficult to assess whether the trade was emanating from an LA or an ST because the rule rejects trades that are modestly profitable. That is, as the broker increases the value of $\xi$ and rejects more trades, she is risking rejecting trades from STs and not only those of the LAs.

\section{Optimal Spread for a Slow Trader and Value of Order Flow}\label{sec: optimal spread for ST}

As seen in the last section, if the venue selects a cutoff level $\xi$, then there is a unique optimal spread $\Delta^*$ which earns the risk-neutral broker zero-expected profit. In other words, there is an optimal spread such that the brokers' expected revenue  from trading with STs equal the expected losses  from trading with LAs.   Moreover, although the broker is indifferent to the choice of $\xi$,  increasing the cutoff, increases the probability that the  rejected trade stems from an ST and not an LA, see Figure \ref{fig:Upsilon}.

Hence,   what is the optimal cutoff $\xi^*$ and the corresponding optimal  spread? To answer this question, we view the problem from the perspective of an ST and the different costs that accrue to the ST. In addition to the expected roundtrip cost $\Omega_{ST}$,   other costs are:  forgone profits which should have accrued to the  ST;  immediacy costs which are high if the ST requires immediate and guaranteed execution -- for example costs that stem from a trading objective that could not be realized (trade could be part of larger operation); and more importantly,  the ST must return to the market to complete the trade which, if executed, is expected to be at a worse price because rejections  occur when prices move in favor of (against) the ST (broker); and, arguably, the ST is exposed to being frontrun.\footnote{Frontrunning is an illegal activity, but FX market participants have argued that Last Look exposes them to frontrunning, see \cite{BoE}. }  Thus  the `effective cost' to the ST is given by
\begin{equation}\label{eqn: omega hat for ST}
\widehat{\Omega}_{ST} = \Omega_{ST} +  C_{ST}(\alpha, \beta, \Delta, \sigma, \theta_{ST}) \,,
\end{equation}
where $\Omega_{ST}$ is the cost to the ST due to the spread and the potential rejection of trades due to Last Look as given in Proposition \ref{thm: cost to ST with last look}, $C_{ST}$ is the additional cost, where $\theta_{ST}$ is a set of idiosyncratic parameters.

We remark that the ST's effective cost is not necessarily lower than the cost that he would incur if trading in a venue without the Last Look option. Thus, depending on the value of the additional cost $ C_{ST}$, the ST will prefer to trade in a venue with Last Look if $\widehat{\Omega}_{ST}< \Delta^0$, where $\Delta^0$ is the spread without Last Look, i.e. $\xi=-\infty$. If the proportion $\alpha $ of LAs in the market is not too  large, so that we can use the simpler expression for the spread without Last Look in Proposition \ref{prop: asymp slow spread no last look}, then  STs  prefer venues with Last Look as long as their  effective costs  are such that
\begin{equation}
\widehat{\Omega}_{ST} < \sqrt{\frac{2}{\pi}}\,\alpha\,\sigma\,.
\end{equation}

Moreover, when the ST prefers venues with Last Look, our results also help to determine  the rejection threshold which minimizes the ST's effective cost.
Figure \ref{fig:hat Omega ST} shows the ST's effective cost with
\begin{equation}\label{eqn: particular Cost}
C_{ST}(\alpha, \beta, \Delta, \sigma, \theta_{ST})=\delta\,(1-\Psi_{ST})\,,
\end{equation}
where $\delta = 0.5\,|\xi|$, and recall  $\Psi_{ST}$ is the probability that the ST's trade is accepted and given in \eqref{eqn: Psi ST},  $\beta=0.8$,  and $\alpha = 0.15$. This choice of $\delta$ is such that every time the ST's profitable trade is rejected, he imputes a cost of half the broker's rejection threshold which is less than half of the forgone profits. For this choice of parameters it is clear that there is an optimal spread where the costs to the ST are minimized. The ST's effective cost is minimized at   $\xi^*/\sigma = -2.49$ which corresponds to  an optimal spread of  $\Delta^*]/\sigma=0.065$, (one can also trace this optimal spread by looking at the left panel in Figure \ref{fig:Delta Opt No Lastlook}). Finally, this spread  is about $50\%$ of the   spread that the broker  charges in the absence of the Last Look option,  which is $\Delta/\sigma=0.12$  (see spreads  as $\xi/\sigma$ goes to $-\infty$ in the left panel of Figure \ref{fig:Delta Opt}).

\begin{SCfigure}[20][!t]
\hspace{5em}
\includegraphics[width=0.45\textwidth]{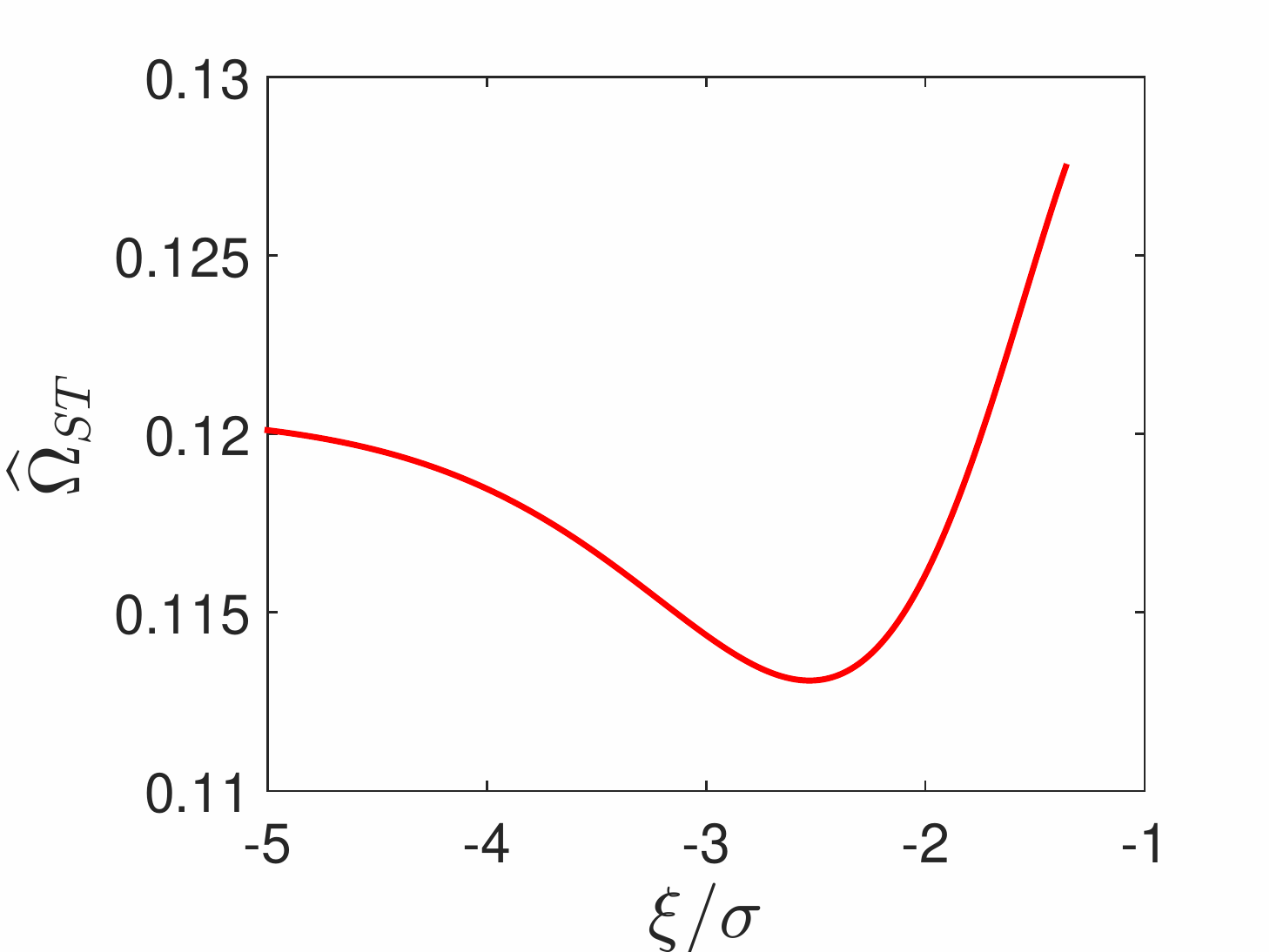}
\hspace{2em}
\caption{The effective cost to the ST accounting for the cost of rejected trades. $\beta=0.8$, $\delta=0.5\,|\xi|$, $\rho=0.5$, $\alpha=0.15$.   \vspace{6em} }\label{fig:hat Omega ST}
\end{SCfigure}

In our model we assume that the broker does not know the type of trader she is facing, but when FX transactions are over-the-counter (instead of an ECN where the counterparty is anonymous) the broker has more information about the identity and strategies of her counterparties. For example, the broker might know if she is facing a trader who executes latency arbitrage trades and she is still willing to trade (and reject) some of the trades. LAs may  also be considered informed traders so the broker benefits from observing the order flow from informed traders. Recall that liquidity providers make prices to their over-the-counter clients and also post quotes on other venues and ECNs. Thus, observing order flow from informed traders is valuable. We could include this in our model in the same way that we included the additional cost that the STs incur, but in this case the broker imputes a positive revenue to executing trades with LAs. Thus, the broker's effective losses to LAs are
\begin{equation}\label{eqn: omega hat for LA}
\widehat{\Omega}_{LA} = \Omega_{LA} - C_{LA}(\alpha, \beta, \Delta, \sigma, \theta_{LA})\,,
\end{equation}
where $C_{LA}\geq 0$ is the benefit that the broker imputes to learning from LAs' order flow.


\section{Asymptotic expressions: Spread, profit, and cost}

When the proportion of LAs in the market is small, the expressions for: the optimal spread (with Last Look),  expected profit and cost of a round-trip trade for LAs and STs,  can be approximated to first order. Later, in Section \ref{sec: equilibrium venues} we employ these expressions to show the equilibrium quantities when there are multiple venues.

\begin{proposition}\textbf{Asymptotic Optimal Spread with Last Look.}\label{prop: asymptotic spread}
When the proportion of LAs trading in the market is small, the asymptotic solution of the optimal spread is given by
\begin{equation}
\frac{\Delta^*}{\sigma} = \tilde\Delta_0 + \tilde\Delta_1\,\alpha + o(\alpha)\,,
\end{equation}
where
\begin{equation}
\tilde\Delta_0 = -\frac{ (1-\beta)\,\sqrt{\frac{1+\rho}{2}}\;\phi\left(\frac{\txi}{\sqrt{2(1+\rho)}}\,\right) }{\beta\,\Phi\left(-\txi\right)+(1-\beta)\,\Phi\left(-\frac{\txi}{\sqrt{2(1+\rho)}}\right)}\,,
\label{eqn: asymp spread Delta_0}
\end{equation}
and
\begin{equation}
\tilde\Delta_1 = 2 \frac{ B(\tilde\Delta_0)- \tilde\Delta_0\,A(\tilde\Delta_0) }{\beta\,\Phi\left(-\txi\right)+(1-\beta)\Phi\left(-\frac{\txi}{\sqrt{2(1+\rho)}}\right)}\,,
\label{eqn: asymp spread Delta_1}
\end{equation}
and $A(\cdot)$ and $B(\cdot)$ are defined in \eqref{eqn:LA Prob Exec I} and \eqref{eqn:LA Prob Exec II}, respectively.
\end{proposition}
\begin{proof}
See \ref{sec: proof of prop: asymptotic spread}.
\end{proof}

\begin{proposition}\textbf{Asymptotic Cost to STs.}
\label{prop: asymptotic omega ft}
When the proportion of LAs trading in the market is small, the broker sets spreads to make zero net profit according to \eqref{eqn: optimal spread balancing equation}, and $C_{ST}$ is as in \eqref{eqn: particular Cost}, the expected (asymptotic) costs of a round-trip trade to STs are
\begin{equation}
\widehat\Omega_{ST} = \eta_0\,\sigma + \eta_1\,\sigma \,\alpha + o(\alpha)\,,
\end{equation}
where
\[
\eta_0 = \frac{\delta}{\sigma}\,(1-\Psi_{ST})\,, \qquad \text{and} \qquad \eta_1=2\,(B(\tilde\Delta_0)-\tilde\Delta_0\,A(\tilde\Delta_0))\,\,.
\]
\end{proposition}
\begin{proof}
See \ref{sec: proof of prop: asymptotic omega ft}.
\end{proof}

\begin{proposition}\textbf{Asymptotic Profit to LAs.}
\label{prop: asymptotic omega la}
When the proportion of LAs trading in the market is small, the expected (asymptotic) profit of a round-trip trade  to LAs is
\begin{equation}
\Omega_{LA}= \gamma_0\,\sigma + \gamma_1\,\sigma\,\alpha + o(\alpha)\,,
\end{equation}
where
\[
\gamma_0 = 2\,\left( B(\tilde\Delta_0) - \tilde\Delta_0\,A(\tilde\Delta_0)\right)\,, \qquad
\gamma_1 = 2\,\left( B'(\tilde\Delta_0) - A(\tilde\Delta_0)- \tilde\Delta_0\,A'(\tilde\Delta_0)\right)\,,
\]
$A(\cdot)$ and $B(\cdot)$ are as in \eqref{eqn:LA Prob Exec I} and \eqref{eqn:LA Prob Exec II} respectively, and $A'(\cdot)$ and $B'(\cdot)$ denote derivatives w.r.t. $\tDelta$:
\begin{equation}
A'(\tDelta) = -\sqrt{1-\rho^2}\,\phi(\tDelta)\,\Phi\left(-\frac{\txi}{\sqrt{2(1+\rho)}}\right)\,,
\end{equation}
and
\begin{equation}
\begin{split}
B'(\tDelta) =&\, -\left\{
\frac{1+\rho}{\sqrt{1-\rho^2}}\,
\phi\left(-\frac{\txi+(1+\rho)\tDelta}{\sqrt{1-\rho^2}}\right)
+\tDelta \; \Phi\left(-\frac{\txi+(1+\rho)\tDelta}{\sqrt{1-\rho^2}}\right)
\right\}\phi(\tDelta) \,
\\
&\,
+ \,\sqrt{\frac{1+\rho}{1-\rho}} \;\phi\left(\frac{\txi}{\sqrt{2(1+\rho)}}\right) \;\phi\left( -\frac{\txi+2\,\tDelta}{\sqrt{2(1-\rho)}} \right)\,.
\end{split}
\end{equation}
\end{proposition}
\begin{proof}
See \ref{sec: proof of prop: asymptotic omega la}.
\end{proof}

\section{Equilibrium: trading in multiple venues}\label{sec: equilibrium venues}

When there is more than one venue to trade, STs  will migrate to the one where the expected losses of a round-trip trade are lowest, and LAs will migrate to the one where the  expected gains are highest. Thus, the market is in  equilibrium  when there are no incentives for either type of trader to migrate to a different venue.  On the other hand, brokers have no preference for a particular venue because spreads are set so that expected profits are zero. Moreover, recall that we assume that brokers do not pay any costs from entering/exiting a venue.

Assume there are $n$ venues to trade and each venue chooses a rejection threshold $\xi_i$,  $i={1,\,2,\,\cdots\, n}$.  Brokers and market makers in all venues are as the one described above: risk-neutral and quote spreads using the zero expected profit condition so that losses to LAs are recovered from STs, i.e.  in each venue spreads are set so that $\alpha\,\Omega_{LA} = (1-\alpha)\, \Omega_{ST}$. When traders switch between venues they incur a fixed cost denoted by $c\geq 0$. This includes customized connection costs and the costs associated with building  a relationship with the broker in  the over-the-counter  FX market.

\begin{definition}\textbf{Equilibrium Across Venues.}\label{defn: equilibrium}
Let $c$ denote the fixed migration costs between  venues and $\xi_i$ denote the rejection threshold of venue $i$. In a market with $n$ venues, an equilibrium (no incentives to migrate)  are pairs  $(\alpha_i, \,\Delta_i)$ for $i={1,\,2,\,\cdots\, n}$ such that all of the following are (simultaneously) satisfied:
\begin{subequations}\label{eqn: equi condition META}
\begin{equation}\label{eqn: equi condition ST}
\left|\,\widehat\Omega_{ST}^i(\alpha_i, \,\Delta_i) -  \widehat\Omega_{ST}^j(\alpha_j, \,\Delta_j)  \,\right| \leq c\,,
\end{equation}
and
\begin{equation}\label{eqn: equi condition LA}
\left|\,\Omega_{LA}^i(\alpha_i, \,\Delta_i) -  \Omega_{LA}^j(\alpha_j, \,\Delta_j)  \,\right| \leq c\,,
\end{equation}
for $i\neq j$, and
\begin{equation}
(1-\alpha_i)\,\Omega_{ST}^i(\alpha_i,\Delta_i) = \alpha_i\,\Omega_{LA}^i(\alpha_i,\Delta_i)
\label{eqn: equi condition spread}
\end{equation}
for all $i$, where superscripts label the venue.

In addition, the population preserving relationships  must be satisfied:
\begin{align}
\alpha_i =&\, \frac{N_{LA}^i}{N_{LA}^i + N_{ST}^i}\,, \\
N_{LA} =&\, \sum^n_i N_{LA}^i \,,  \\
N_{ST} =&\, \sum^n_i N_{ST}^i\,,%
\end{align}
and the constraints
\begin{equation}
N_{ST}^i \,, \, N_{LA}^i\, \ge 0\,.
\end{equation}
\label{eqn: Equilibrium Exact}
\end{subequations}
\end{definition}

In this definition we assume that traders decide to migrate if the gains from one trade exceed the fixed migration costs. An alternative is to calculate the migration gains employing the number of transactions that the trader expects to execute in the new venue, in which case the left-hand side of inequalities \eqref{eqn: equi condition ST}, \eqref{eqn: equi condition LA} is premultiplied by the expected number of trades.

\subsection{Equilibrium across two FX trading venues}

Assume there are two venues which employ rejection thresholds $\xi_1$ and $\xi_2$. Let $N_{LA}$ and $N_{ST}$ denote the total number of LAs and STs in the market. These traders   choose which venue to trade in and decide to migrate if they are better off in the other venue.  As discussed above, the venues are in equilibrium if the expected costs for STs and expected profits for LAs, net of the migration cost $c$,  are the same across both venues -- so the marginal trader, whether ST or LA, has no incentives to migrate.

To obtain the equilibrium region we proceed as follows. For each venue we find the pairs $(\alpha_i,\,\Delta_i)$ such that STs do not have an incentive to migrate and the region where LAs do not have an incentive to migrate. That is, we find   the regions where \eqref{eqn: equi condition ST} and \eqref{eqn: equi condition LA} (together with  the population constraints and the brokers' zero expected profit condition) both hold.  Thus, the intersection between these two regions define the equilibrium where traders do not migrate to the other venue.

To obtain the regions where the two types of traders are indifferent between the two venues, we can use the closed-form formulae derived above for the optimal spread, LA's expected profits and ST's expected costs. Alternatively, if the proportion of LAs in each venue is small,  we can employ the expressions in Propositions \ref{prop: asymptotic spread}, \ref{prop: asymptotic omega ft}, and \ref{prop: asymptotic omega la}. Either approach will result in approximately the same equilibrium region. There are two advantages to employing the small $\alpha$ approximations: i) computations are extremely fast, ii) we can characterize the equilibrium region in closed-form. For the parameters we used, there is no discernable difference between the exact and approximate equilibrium regions, nor the optimal spreads implied by them.

Figure \ref{fig: equi region} shows the equilibrium region for Venue 1, when migration costs are $c=0.05$ (left-hand panel), and $c=0.025$ (right-hand panel).  The additional costs incurred by the STs are as in \eqref{eqn: particular Cost} with $\delta = 0.5\,|\xi|$. The other parameters are:  total number of LAs   $N_{LA} = 200$, total number of STs  $N_{ST} = 800$, rejection threshold in Venue 1 is fixed at $\xi_1=-3.5$, and there is no Last Look in Venue 2.  The equilibrium region is obtained using the small $\alpha$ formulae.
\begin{figure}[!t]
\begin{center}
\includegraphics[width=0.45\textwidth]{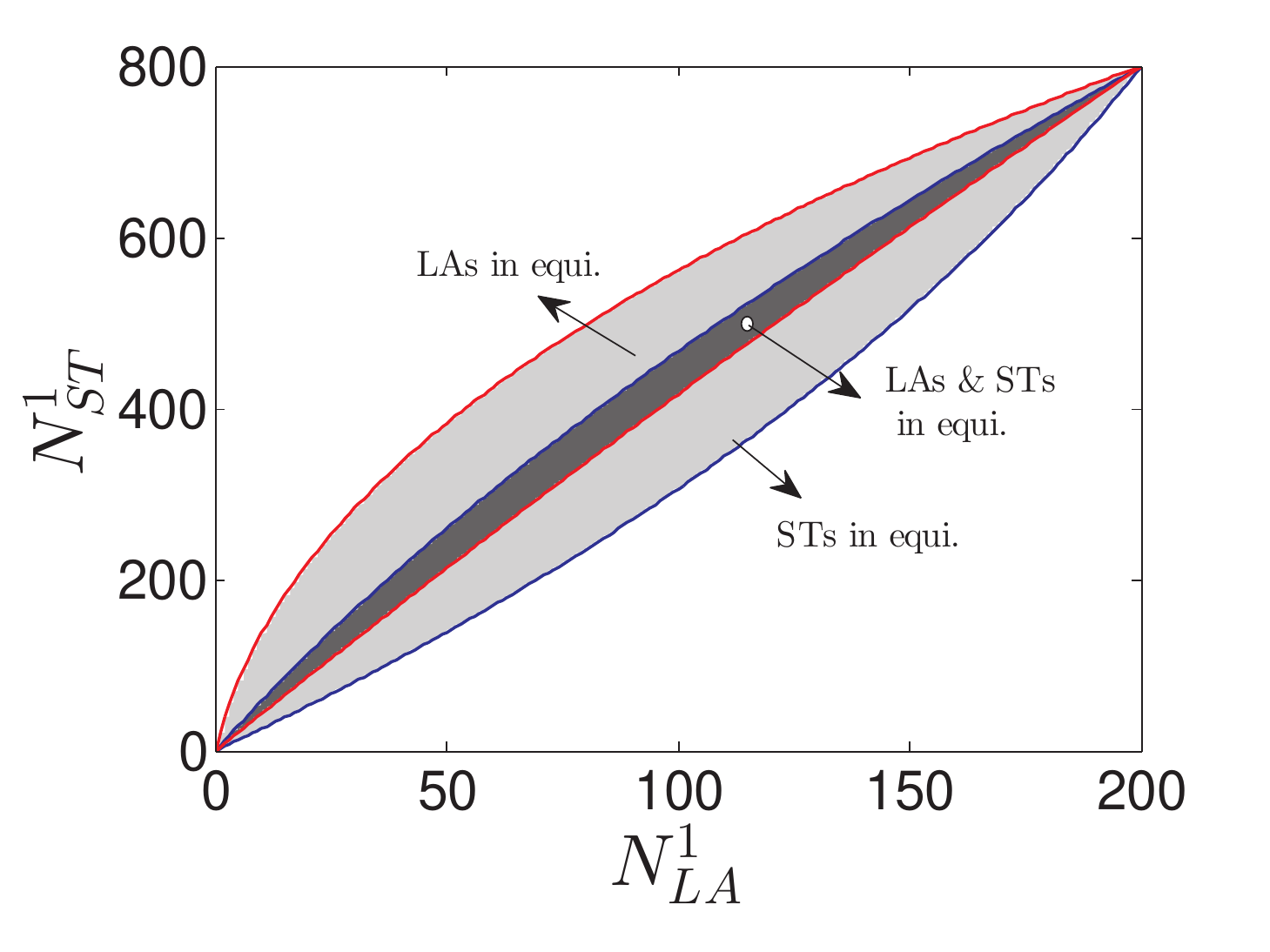}
\includegraphics[width=0.45\textwidth]{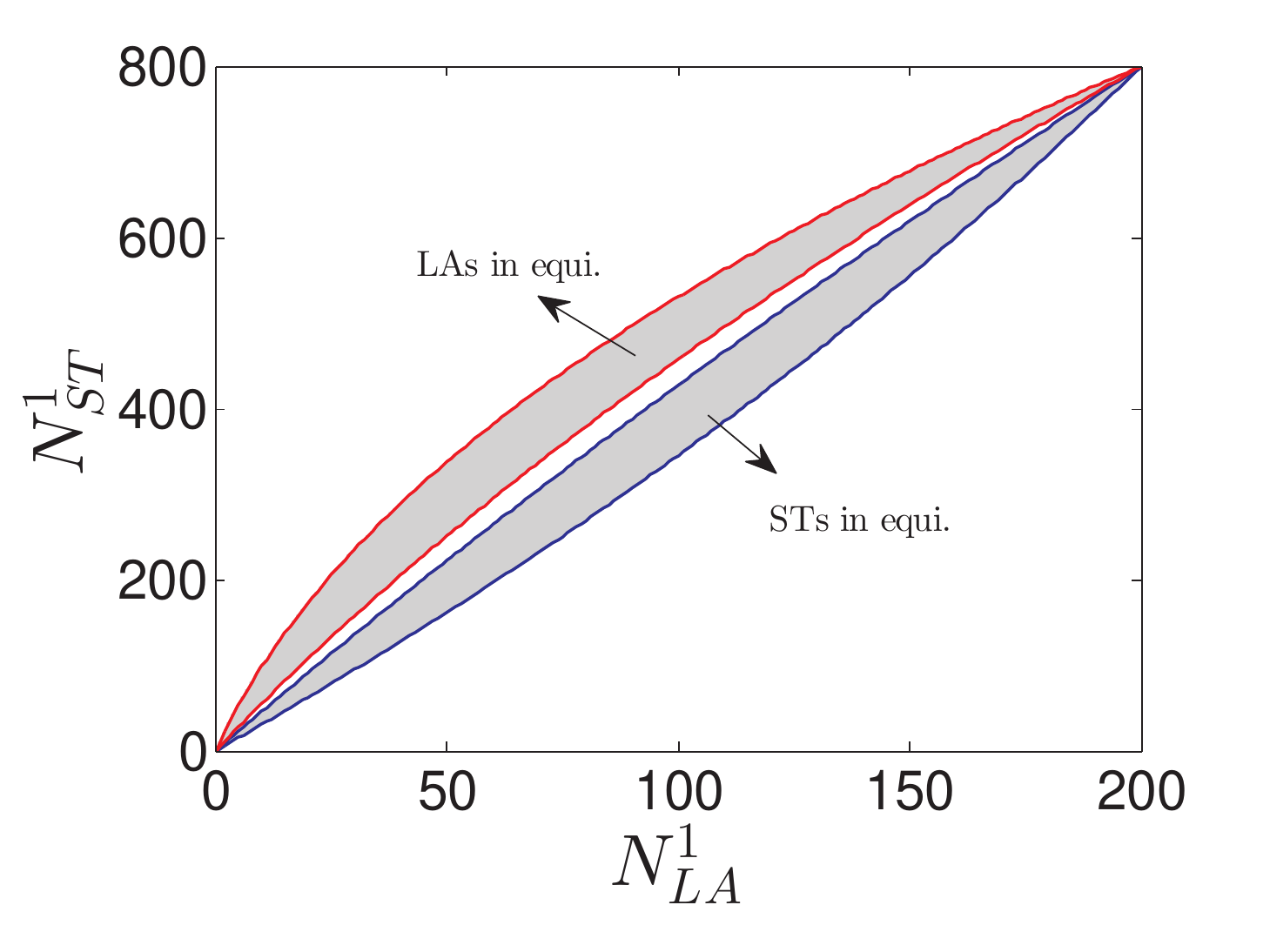}
\end{center}
\caption{Equilibrium region (dark gray) in Venue 1 with $\xi_1= -3.5$ and Venue 2 (not shown) has no Last Look. Left panel migration cost is $c=0.05$, and right panel $c=0.025$. The other parameters are $\sigma=1$, $\beta=0.8$, $\rho=0.5$, and $\delta = 0.5\,|\xi|$. Red lines bound the equilibrium region for LAs, blue lines bound the equilibrium region for STs. }\label{fig: equi region}
\end{figure}

In the left panel of the figure the equilibrium region  (dark gray) clearly shows that both venues can co-exist but the number of traders that each venue supports can vary from very few traders to nearly all traders. At all points in this equilibrium, neither STs nor LAs find it optimal to migrate to the other venue. The  region between the blue lines (which includes the dark gray region) is where STs are indifferent between the two venues. Similarly, the region between the red lines is where LAs are indifferent between the two venues.   Here we assume that venues can survive with little order flow or that there is no value to brokers from observing flow. In more realistic scenarios, where brokers impute   value to order flow (so their profit  function is different from the one assumed above), these results will very likely differ -- see discussion leading to equation \eqref{eqn: omega hat for LA}.

If the market is at a point outside the equilibrium region there are incentives to flow between the two venues until it is suboptimal for any type of trader to migrate. The path that traders take from disequilibrium to an equilibrium depends on how quickly they spot, and can act on,  better opportunities. Note that as soon as one trader changes venue, the proportion of LAs in both venues changes and brokers must adjust the quoted spreads to break-even. These changes in both quoted spreads and proportion of LAs, affect the profitability of round-trip trades for LAs and the costs borne by STs, so both types reassess whether they should remain in their current venue or migrate to the other one. 

Another interesting feature to observe is that the equilibrium region shrinks as migration costs to trades become smaller. In the right panel of the figure the migration cost is $c=0.025$  and we observe that the market \emph{cannot} reach an equilibrium. Clearly,  in markets where migration is costly there are less incentives for traders to switch venues. Similarly, in markets where traders can easily switch venues will show more traffic of traders between them because traders can exploit any discrepancy, however small, between the costs and profits of trading in the two venues.

\subsection{Analytical Characterization of Equilibrium Region}

When the asymptotic forms of the value to LAs and costs to STs provided in Propositions \ref{prop: asymptotic omega ft} and \ref{prop: asymptotic omega la} are used, we can characterize the equilibrium region for the two-venue case in  a compact form. Both constraints \eqref{eqn: equi condition ST} and \eqref{eqn: equi condition LA} reduce to the same form and only differ in the coefficients that appear. Hence, we focus only on rewriting \eqref{eqn: equi condition ST} subject to the condition \eqref{eqn: equi condition spread} and the population preserving constraints.

First, using Proposition \ref{prop: asymptotic omega ft}, \eqref{eqn: equi condition ST} subject to the broker setting the spread to make zero expected profits, i.e. that \eqref{eqn: equi condition spread} is satisfied, reduces to
\[
\left| H_0 + \eta_1^1\,\alpha^1 - \eta_1^2\,\alpha^2 \right| \le c\,,
\]
where $H_0=\eta_0^1-\eta_0^2$. Imposing the population constraint further implies that
\[
\left| H_0 + \eta_1^1\,\frac{x}{x+y} - \eta_1^2\,\frac{M-x}{N-(x+y)} \right| \le c \,,
\]
where $x$ and $y$ represent the number of LAs and STs, respectively, in Venue 1, $N$ is total population size, $M$ is the total number of LAs, and the constant $H_0=\eta_0^1-\eta_0^2$. The population constraints also impose the conditions $0\le x\le M$ and $0 \le x+y\le N$ which implies that the numerator and denominator of each of the fractions appearing above are all non-negative. We can rewrite this inequality as the following pair of inequalities
\[
\begin{split}
H_0 + \eta_1^1\,\frac{x}{x+y} - \eta_1^2\,\frac{M-x}{N-(x+y)}  \lesseqqgtr \pm c\,.
\end{split}
\]
Multiplying by $(x+y)(N-(x+y))$, which is positive due to the population constraints, we obtain, after some tedious algebra,
\begin{equation}
\begin{split}
(\eta_1^2-\eta_1^1-\zeta_\pm)\, x^2
&+ (\eta_1^2-\eta_1^1-2\,\zeta_\pm)\,x\,y
-\zeta_\pm\,y^2 \\
&+ ((\zeta_\pm+\eta_1^1)\,N-\eta_1^2\,M)\,x
+ (\zeta_\pm\,N-\eta_1^2 \,M)\,y \lesseqqgtr 0\,,
\end{split}
\label{eqn: conic}
\end{equation}
where the constants
\[
\zeta_\pm=H_0\mp c\,.
\]

If the inequalities above  are replaced by equality, then \eqref{eqn: conic} represent conic sections. A standard result shows that, after a rotation and a translation, there are three cases (when non-degenerate). Letting
$\omega_\pm = B^2 - 4\,A\,C$, where $A$, $B$ and $C$ are the coefficients of $x^2$, $xy$ and $y^2$, respectively, then if
\begin{enumerate}
\item $\omega_\pm <0$,  the conic section is an ellipse,

\item $\omega_\pm >0$,  the conic section is a hyperbola, and

\item $\omega_\pm =0$,  the conic section is a parabola.

\end{enumerate}
From \eqref{eqn: conic}, we see that
\[
\omega_\pm=\left(\eta_1^2-\eta_1^1-2\,\zeta_\pm\right)^2+4\,\left(\eta_1^2-\eta_1^1-\zeta_\pm\right)\,\zeta_\pm
= \left(\eta_1^2-\eta_1^1\right)^2 \ge 0\,,
\]
hence the conics are rotated and translated hyperbolae or parabolas. For example, parabolas appear  when $\eta_1^2 = \eta_1^1$ -- one such case is when the two venues are identical. Moreover, by direct substitution into \eqref{eqn: conic}, we see that the hyperbolae go through the origin $(x,\;y)=0$ as well as the corner $(x,\;y) = (M, \;N)$ -- i.e. either there are no traders in Venue 1 (and no flow into that venue), or all traders are in Venue 1 (and there is no flow out of that venue).

%

\subsection{Path to equilibrium between two venues}\label{subsec: reaching equi}

Here we illustrate  how traders migrate between two venues  until they reach an equilibrium.  We use the closed-form formulae derived above to obtain the equilibrium pairs $(\alpha,\, \Delta)$.  We assume that there are two venues where the proportion of LAs and quoted spreads are such that in each individual venue the broker makes zero net expected profits from trading, however, there may be incentives for traders to migrate.  We assume that traders, whether an LA or an ST, move between venues at a rate  proportional to the gain in expected value, after accounting for switching costs, they receive from making the migration only if these gains are positive. To this end, let $n_{LA}(t)$ and $n_{ST}(t)$ denote the number of LAs and STs in Venue $1$, and let $N_{LA}$, $N_{ST}$, and $N$ denote the total number of LAs, STs, and total participants in the market, we assume the dynamic flow
\begin{subequations}\label{eqn: ODE flow asymp}
\begin{align}
\begin{split}
\frac{dn_{LA}}{dt} =&\, \frac{\kappa_{LA}}{\sigma} \left\{
\left(
  \Omega_{LA}^1\left(\tfrac{n_{LA}}{n_{ST}+n_{LA}}\right)
- \Omega_{LA}^2\left(\tfrac{N_{LA}-n_{LA}}{N-(n_{ST}+n_{LA})}\right)
-c_{LA}\right)_+ \! \mathds 1_{\{n_{LA}<N_{LA}\}} \right. \\
& \qquad - \left.
\left(
  \Omega_{LA}^2\left(\tfrac{N_{LA}-n_{LA}}{N-(n_{ST}+n_{LA})}\right)
- \Omega_{LA}^1\left(\tfrac{n_{LA}}{n_{ST}+n_{LA}}\right)
-c_{LA}\right)_+ \! \mathds 1_{\{n_{LA}>0\}}
\right\}\,,
\end{split}
\\
\begin{split}
\frac{dn_{ST}}{dt} =&\, \frac{\kappa_{ST}}{\sigma} \left\{
\left(
  \wOmega_{ST}^1\left(\tfrac{n_{LA}}{n_{ST}+n_{LA}}\right)
- \wOmega_{ST}^2\left(\tfrac{N_{LA}-n_{LA}}{N-(n_{ST}+n_{LA})}\right)
-c_{ST}\right)_+ \! \mathds 1_{\{n_{ST}<N_{ST}\}} \right. \\
& \qquad - \left.
\left(
  \wOmega_{ST}^2\left(\tfrac{N_{LA}-n_{LA}}{N-(n_{ST}+n_{LA})}\right)
- \wOmega_{ST}^1\left(\tfrac{n_{LA}}{n_{ST}+n_{LA}}\right)
-c_{ST}\right)_+ \! \mathds 1_{\{n_{ST}<0\}}
\right\}\,,
\end{split}
\end{align}
\end{subequations}
where we have suppressed the explicit dependence on $t$ for compactness, the superscripts label the venues, recall that $(x)_+=\max(x, \,0)$, $\kappa_{LA},\kappa_{ST}>0$ are constants which transform the migration gains into rates, and $c_{LA}, c_{ST}\ge 0$ are the costs of switching from one venue to the other.

Throughout we assume that all market makers know exactly the parameters in the model and react immediately to the flow of traders, however, in reality this information would be corrupted by noise. To account for this, we could add in Brownian motion components to \eqref{eqn: ODE flow asymp}, which changes the ordinary differential equations (ODEs) into stochastic differential equations and no equilibria would exist, instead the flow would approach the noise free equilibrium regions, but fluctuate around them.

The above equations define  a system of coupled non-linear  ODEs and we cannot hope to solve them in general.
There are, however, a few simple features of this dynamic flow that we can glean. In the equilibrium region, the right-hand sides of \eqref{eqn: ODE flow asymp} are both zero and there is no migration between venues. In the region where LAs have no incentive to migrate, but the STs do, (e.g., the region between the red lines in Figure \ref{fig: equi region}), then there is flow in only $n_{ST}$. In the region where STs have no incentive to migrate, but the LAs do, (e.g., the region between the blue lines in Figure \ref{fig: equi region}), then there is flow in only $n_{LA}$.

To illustrate how the market reaches an equilibrium we first look at an example where there are two venues that start at a particular point outside the equilibrium region and traders migrate between venues until an equilibrium point is reached. After this example we examine the general case by considering all possible starting points and employ the coupled system of ODEs to show the path that traders take until an equilibrium is reached.

Assume that Venue 1 fixes a rejection threshold, Venue 2 does not have the Last Look option,  and each venue starts with a given  number of LAs and STs.\footnote{Traders who know $\sigma$, $\beta$,  $\xi$,  and $\rho$  can infer the proportions of LAs, in each venue, from posted spreads. } In our first example migration between venues is sequential: at every step, one trader of each type may migrate to the other venue. Brokers  and traders can always observe the  number of LAs and STs trading in the venue. Thus, immediately after migration,  brokers  in both venues calculate the new break-even spreads, traders also calculate the new  expected costs and profits of round-trip trades and reassess whether they should stay or migrate, and so on. This is repeated until there are no incentives to migrate.

Moreover, at the beginning, in Venue 1 there are  $N^1_{LA}=125$   and $N^1_{ST}=375$, so $\alpha_1 = 25\%$. And the starting point in Venue 2 is  $N^2_{LA}=75$ and $N^2_{ST}=425$ so $\alpha_2  = 15\%$. Recall that Venue 2 does not have the Last Look option. Table \ref{tab: fixed xi no last look} shows the starting and equilibrium configuration for two examples: in the left-hand panel Venue 1 employs a rejection threshold $\xi_1=-4$ and in the right-hand panel it employs a  stricter rejection threshold of $\xi_1 = -3.5$.

\begin{table}[h!]
\begin{center}
\begin{minipage}{0.48\textwidth}

\scriptsize

\begin{tabular}{rrrrrr}
\hline
\hline
\multicolumn{6}{c}{$\xi_1=-4$, $\xi_2=-\infty$ (no Last Look)} \bigstrut\\
\hline
      & \multicolumn{2}{c}{Initial} & \multicolumn{1}{c}{} & \multicolumn{2}{c}{Final} \bigstrut\\
\cline{2-3}\cline{5-6}    $\alpha$ & 25.0\% & 15.0\% &       & 21.1\% & 19.0\% \bigstrut[t]\\
    $N$   & 500   & 500   &       & 469   & 531 \\
    $N_{LA}$ & 125   & 75    &       & 99    & 101 \\
    $N_{ST}$ & 375   & 425   &       & 370   & 430 \\
    $\Delta^*$ & 0.19  & 0.12  &       & 0.16  & 0.15 \\
    $\widehat\Omega_{ST}$ & 0.20  & 0.12  &       & 0.17  & 0.15 \\
    $\Omega_{LA}$ & 0.58  & 0.68  &       & 0.61  & 0.65 \bigstrut[b]\\

\hline
\hline
\end{tabular}%
\end{minipage}
\begin{minipage}{0.48\textwidth}

\scriptsize

\begin{tabular}{rrrrrr}
\hline
\hline
\multicolumn{6}{c}{$\xi_1=-3.5$,  $\xi_2=-\infty$ (no Last Look)} \bigstrut\\
\hline
      & \multicolumn{2}{c}{Initial} & \multicolumn{1}{c}{} & \multicolumn{2}{c}{Final} \bigstrut\\

\cline{2-3}\cline{5-6}    $\alpha$ & 25.0\% & 15.0\% &       & 18.6\% & 21.1\% \bigstrut[t]\\
    $N$   & 500   & 500   &       & 456   & 544 \\
    $N_{LA}$ & 125   & 75    &       & 85    & 115 \\
    $N_{ST}$ & 375   & 425   &       & 371   & 429 \\
    $\Delta^*$ & 0.18  & 0.12  &       & 0.13  & 0.17 \\
    $\widehat\Omega_{ST}$ & 0.19  & 0.12  &       & 0.14  & 0.17 \\
    $\Omega_{LA}$ & 0.55  & 0.68  &       & 0.59  & 0.64 \bigstrut[b]\\

\hline
\hline
\end{tabular}%


\end{minipage}

\end{center}

\caption{Equilibrium across venues, fixed rejection thresholds and varying spreads,  $\beta=0.8$, $\sigma= 1$, $\rho = 0.5$, $\delta = 0.5\,|\xi|$,  $c=0.05$. } \label{tab: fixed xi no last look}

\end{table}

The two panels in the table show how the market reaches an equilibrium where a venue without Last Look coexists with one where brokers have the right to reject trades. With the assumption that only one trader of each type may migrate at each time-step, we see that equilibrium is reached where the proportion of LAs in each venue is close to 20\%, despite the fact that the starting points were 25\% and 15\%. We observe that in the left-hand side panel, the lowest expected cost of a round-trip for an ST is in the venue without Last Look, but in the right-hand panel STs are better off in the venue with the Last Look option. Moreover, it is also interesting to observe the equilibrium spreads: in the left panel, the venue without Last Look   quotes a \emph{tighter} spread than the venue with Last Look  -- whereas in the right panel we see that the venue with Last Look quotes a tighter spread than the venue  without  Last Look.

\begin{figure}[!t]
\begin{center}
\includegraphics[width=0.45\textwidth]{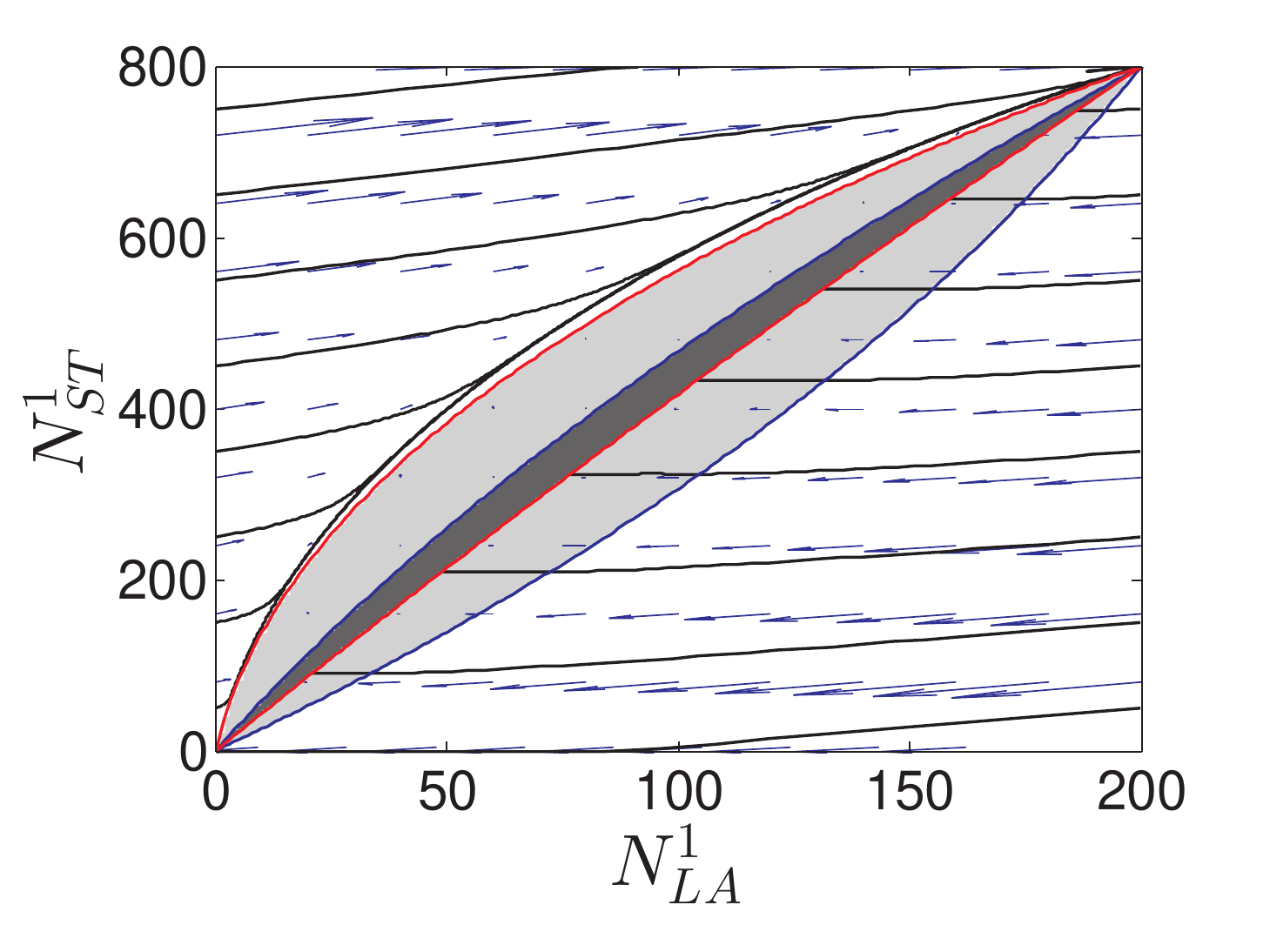}
\includegraphics[width=0.45\textwidth]{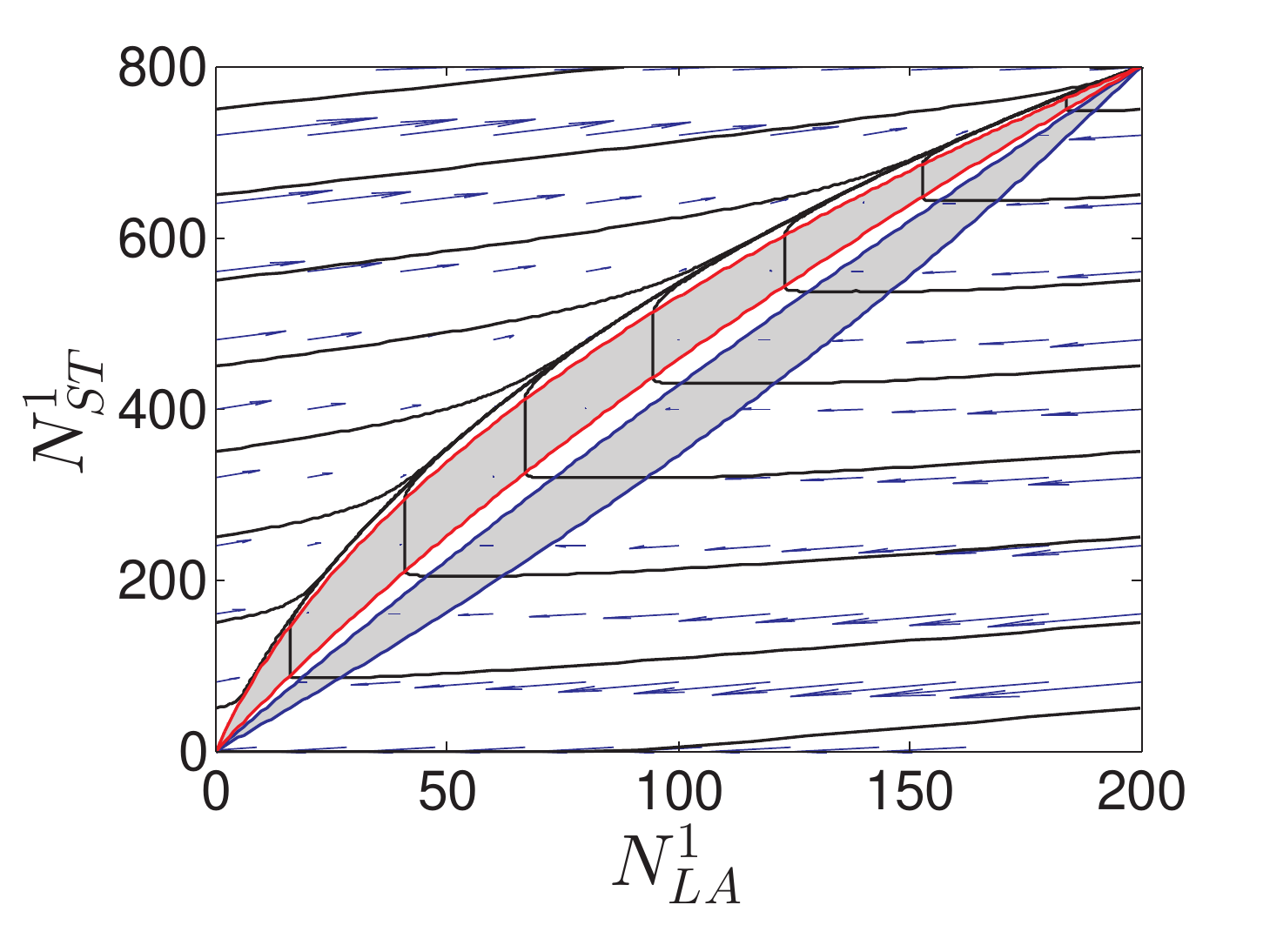}
\end{center}
\caption{Equilibrium region (dark gray) in Venue 1 with $\xi_1= -3.5$ and Venue 2 (not shown) without Last Look, and  $c=0.05$ and $0.025$ in the left and right panels. The other parameters are $\sigma=1$, $\beta=0.8$, $\rho=0.5$, $\kappa_{LA} =40$, $\kappa_{ST} =20$,  and $\delta = 0.5\,|\xi|$. Black lines indicate the migration of traders. Blue arrows  indicate the direction of the migration. Red lines bound the equilibrium region for LAs, blue lines bound the equilibrium region for STs. \label{fig: equi region flow}}
\end{figure}
Now we examine the general case where we consider all possible starting points in each venue and use the migration dynamics described by \eqref{eqn: ODE flow asymp} to show the path to equilibrium.  LAs are faster than other market participants, so they migrate between venues at a faster rate, i.e. $\kappa_{LA}>\kappa_{ST}$, and in particular we use $\kappa_{LA} =40$, $\kappa_{ST} =20$.  Figure \ref{fig: equi region flow} shows the migration paths seen in Venue 1  when migration costs are $c=0.05$ (left panel) and $c=0.025$ (right panel).  Figure \ref{fig: equi region flow} is the same as Figure \ref{fig: equi region} but it also shows, in black lines, the migration path of traders, and the blue arrows show the direction of the migration. Moreover, recall that the region between red lines is where the LAs do not have incentives to migrate to the other venue, and the region between blue lines is where the STs do not have incentives to migrate.

In the left panel, where migration costs are $c=0.05$, we observe that when the starting point is in the `lower triangular' white area, both STs and LAs have incentives to migrate to Venue 2 (they are better off in Venue 2 which has no  Last Look) and equilibrium is eventually reached. In contrast, for any starting point in the `upper triangular' white area, the equilibrium point is where Venue 1 attracts \emph{all}  the traders in the market -- the Venue without Last Look loses all flow to Venue 1.

The picture in the right-hand panel shows that when migration costs are low, so that there is no equilibrium region as already discussed above, traders migrate to two corner solutions: all traders are in Venue 1 or are in Venue 2, i.e. only one FX venue survives in the marketplace. Note that only when the starting point is in the lower triangular region and the number of STs is small, do we see that all traders exit Venue 1 and prefer to trade in Venue 2 without Last Look. In all other cases,
migration occurs until \emph{all} traders leave Venue 2 in favor of Venue 1 with the Last Look option.

Moreover, the migration flows shown in the paths that start in the lower triangular area, and that end up where all traders are in Venue 1, follow an interesting pattern. First we observe that LAs exit Venue 1 and there is not much change in the population of STs. This pattern is seen until the market reaches the region where the STs are in equilibrium (between the blue lines) and at that point STs stop flowing and LAs continue flowing out of Venue 1. Then, the flow reaches the region between the two equilibrium regions. In this region, LAs flow out of Venue 1, while STs flow into Venue 1, causing the flow to get closer to the region where LAs are in equilibrium (between the red lines). Once the flow is in the region where LAs are in equilibrium, they do not flow out of Venue 1 anymore, but STs continue flowing into Venue 1. Then the flow exits the LA equilibrium region and both STs and LAs flow into Venue 1 at a rate which prevents the flow from entering the LA equilibrium region. The reason is that there is migration pressure from STs into Venue 1 within the LA equilibrium region. Interestingly, all these paths lead to an equilibrium where the venue \emph{without} Last Look loses \emph{all} its traders.

Recall that in our model $\alpha$ may also be interpreted as the ratio of latency arbitrage trades to all trades in the market. Thus, the results above may be interpreted as spreads and equilibria across venues attracting trades. For example, an ST could require different immediacy for her trades (which would be  reflected in the effective cost component $C_{ST}$ for each trade) and this determines on which venue the  ST executes the trade. Trades that require guaranteed execution  have a high $C_{ST}$, so are executed on venues with lenient or no rejection threshold.  Finally, although we do not model the flow of market makers between venues, in our set-up  brokers will cease to provide liquidity in venues that disappear and will make markets in other venues. Similarly, venues that do not cease to exist but lose order flow, will also see brokers switch to venues that gained order flow.

\section{Conclusions}\label{sec: conclusions}

We show that risk-neutral market makers or brokers quote tighter spreads to the market when they reject loss-leading trades using the Last Look option.  The Last Look option helps market makers to mitigate their losses to latency arbitrageurs and also reduces the wealth transfer between slow traders and those who arbitrage the market by trading on stale quotes. In our setup the market maker sets spreads so that she makes zero expected profits.

The  Last Look option consists of a time frame and a rejection threshold used by the broker to reject trades ex-post. Since the market maker cannot distinguish the type of trader behind the trades, latency arbitrageur or slow trader, the Last Look option is enforced across all trades. Our results  show that brokers are indifferent between different rejection thresholds because they set optimal spreads so that her losses to latency arbitrageurs are covered by the other traders in the market.

We show how effective is the Last Look option as a function of the rejection threshold which determines the market maker's tolerance to losses on a trade-by-trade basis. When the venue sets a very strict threshold (i.e. any trade that yields a modest profit to the traders is cancelled by the broker), slow traders end up being penalized too often.  On the other hand, if the rejection threshold is set so that only trades which result in large losses to the market maker are rejected, the Last Look option becomes very effective at singling out latency arbitrageurs given the fact that the trade is rejected.

At first sight it  seems that a `relaxed' threshold is better because the probability that a rejected trade came from a latency arbitrageur is higher. The flip side, however,  is that rejection rarely happens, hence losses to latency arbitrage are high, and this results in higher quoted spreads.

Moreover, since  the risk-neutral market maker determines the spread so that expected profits are zero, there is a one-to-one mapping between optimal spreads and rejection thresholds which are set by the venue. Strict thresholds lead to tight spreads, and lenient thresholds lead to large spreads. The extreme case is when the threshold is so lenient that no trades are rejected which is equivalent to trading in a venue without Last Look.   Therefore, when there is only one FX venue, the market maker is indifferent between different levels of the threshold.

Slow traders, on the other hand, are not indifferent between rejection thresholds.  Slow traders benefit from the Last Look option because market makers cap their losses to latency arbitrageurs, but slow traders' most profitable trades are also cancelled. Thus, when slow traders account for forgone earnings  (due to rejected trades), immediacy costs, and the costs from returning to the market to complete the trade,   there is an optimal threshold that minimizes their costs of trading in the venue with Last Look.  If there is only one FX venue, this optimal threshold could be the extreme where market makers \emph{never} reject trades. In other words, depending on: the proportion of latency arbitrageurs acting in the market,  and on the latency of the slow traders, slow traders will seek or avoid venues with Last Look.

When there is more than one FX venue, market makers still post spreads that ensure that losses to LAs are recovered from STs. Competition across venues, however, incentivizes traders to migrate to those where they are better off.  We show that there is an equilibrium region where there are no incentives to migrate. If the market starts outside this region, traders will migrate until an equilibrium is reached. This equilibrium could be one where both venues coexist or one where only one venue survives.

Interestingly, we show that when there are two venues, one with and one without Last Look, the equilibrium reached by the market is chiefly dependent on the proportion of latency arbitrageurs trading in each market. When the no Last Look venue starts with a low proportion of  latency arbitrageurs (i.e. a high proportion of latency arbitrageurs in the Last Look venue) the market reaches an equilibrium where both venues coexist. If the market's starting point, however, is one where the venue with Last Look has a low proportion of latency arbitrageurs, the market reaches an equilibrium  where the venue enforcing Last Look attracts \emph{all} order flow, i.e. only the Last Look venue survives.

\appendix

\section{Proof of Results}

\subsection{Proof of Proposition \ref{thm: loss to LA without lastlook}} \label{sec: proof of loss to LA without lastlook}

The result follows from a straightforward computation of the expectation:
\begin{eqnarray*}
\EE_0\left[(P_1-(P_0+\Delta))_+\right]
&=&
\int_{-\infty}^\infty \left(\sigma\,z - \Delta\right)
\II_{\{\sigma\,z-\Delta > 0\}}\,\frac{e^{-\frac{1}{2}z^2}}{\sqrt{2\pi}}\,dz \\
&=& \left.-\sigma\,\frac{e^{-\frac{1}{2}z^2}}{\sqrt{2\pi}}\right|_{\frac{\Delta}{\sigma}}^\infty
-\Delta \,\Phi\left(\frac{\Delta}{\sigma}\right) \\
&=& \sigma\,\phi\left(\frac{\Delta}{\sigma}\right)-\Delta \Phi\left(\frac{\Delta}{\sigma}\right)\,.
\end{eqnarray*}
\qed

\subsection{Proof of Proposition \ref{prop: unique finite solution for spread no last look}} \label{sec: proof of unique finite solution for spread no last look}

It is easy to check that $f(x) = \phi\left(x\right)-x\,\Phi\left(-x\right)$ is decreasing and convex on $x\in[0,+\infty)$. In particular, we have
\[
f'(x) = -\Phi\left(-x\right) \le 0\,, \qquad
f''(x) = \phi\left(x\right) > 0\,, \qquad
\forall x\ge0\,.
\]
Moreover, $f(0) = \sqrt{\frac{2}{\pi}}$, $f'(0)=-\frac{1}{2}$, $f''(0)>0$ and $\lim_{x\to\infty} f'(x)=0$. Let $g$ denote the line $g(x)=\frac{1-\alpha}{\alpha}\,x$. Clearly, $g(0)<f(0)$ and since $f$ is convex, we must have $f'(x)\le\lim_{x\to+\infty}f'(x)=0$. Together with the fact that $f'(0)=-\frac{1}{2}$ and $f''(0)>0$, we see that there must exist a single intersection point of $f$ and $g$ on $x\in[0,+\infty)$ if and only if the slope of the line $g$ is greater than the asymptotic slope of $f$, i.e. as long as the slope is positive. This occurs if and only if $\alpha\in[0,1]$.\qed

\subsection{Proof of Proposition \ref{prop: asymp slow spread no last look}}
\label{sec: proof of asymp slow spread no last look}

First, rearrange the root of \eqref{eqn:balance no lastlook} to write
\[
2\alpha\left(\phi\left(x^*(\alpha)\right)-x^*(\alpha)\,\Phi\left(-x^*(\alpha)\right)\right) =(1-\alpha)\,x^*(\alpha)\,,
\]
and note the explicit dependence on $\alpha$. Clearly, in the limit $\alpha\downarrow0$, $x^*(\alpha)\downarrow0$. With this observation, next, write
\[
x^*(\alpha) = c\,\alpha+o(\alpha)\,,
\]
and aim to find the constant $c$. Inserting this expansion into the previous expression and expanding in $\alpha$ we find that
\[
2\alpha\left[\left(\phi(0) + c\,\alpha \,\phi'(0)\right)-c\,\alpha
\,\left\{ \Phi(0)-c\,\alpha\,\phi(0)\right\} \right] =(1-\alpha)\,c\,\alpha+o(\alpha)\,,
\]
so that $c = 2 \phi(0)=\sqrt{\frac{2}{\pi}}$ and the result follows. \qed

\subsection{Computation of \eqref{eqn: Omega ST on stale}}
\label{sec: derive Omega ST on stale}

To derive this result, first note that due to symmetry both expectations are equal and therefore,
\[
\Omega_{ST\,|\,\text{stale}}
=-\EE\left[\left(P_1-P_0-\Delta\right)\,\II_{\{P_0-P_2>\xi\}}\right]\,.
\]
Next, separate the two terms in the expectation into two pieces: $A=\EE\left[\left(P_1-P_0\right)\,\II_{\{P_0-P_2>\xi\}}\right]$ and $B=\EE\left[\Delta\,\II_{\{P_0-P_2>\xi\}}\right]$.

The computation of $B$ is straightforward. Since
\[
P_0-P_2=\sigma(Z_1+Z_2) \stackrel{\footnotesize d}{=} \sigma\left(Z_1+\rho\,Z_1+\sqrt{1-\rho^2}\,Z_1^\perp\right)\,,
\]
where $Z_1^\perp$ is a standard normal independent of $Z_1$. Therefore, $P_0-P_2$ is normal with mean $0$ and standard deviation $\sqrt{(1+\rho)^2+(1-\rho^2)}=\sqrt{2(1+\rho)}$, and so
\[
B=\Delta\,\Phi\left(\frac{\xi}{\sqrt{2(1+\rho)}}\right)\,.
\]

Next,  we need the following expectation to compute $A$:
\begin{eqnarray*}
\EE\left[Z_1\,\II_{\{Z_1+Z_2 > c\}} \right]
&=& \EE\left[Z_1\,\II_{\left\{(1+\rho)\,Z_1+\sqrt{1-\rho^2}\,Z_1^\perp > c\right\}} \right] \\
&=& \int_{-\infty}^\infty \int_{a - b\,\zeta_1 }^\infty
\zeta_1\,e^{-\frac{1}{2}\zeta_1^2 - \frac{1}{2} \zeta_2^2}\,\frac{d\zeta_2}{\sqrt{2\pi}}\frac{d\zeta_1}{\sqrt{2\pi}}\,,
\end{eqnarray*}
where $c$ is an arbitrary constant,  $a=\frac{c}{\sqrt{1-\rho^2}}$ and $b=\frac{1-\rho}{\sqrt{1-\rho^2}}$.
Continuing the computation,
\begin{align*}
\EE\left[Z_1\,\II_{\{Z_1+Z_2 > c\}} \right]
&= \int_{-\infty}^\infty
\zeta_1\,\Phi\left(b\,\zeta_1-a\right)\,e^{-\frac{1}{2}\zeta_1^2 }\,\frac{d\zeta_1}{\sqrt{2\pi}} \\
&= \frac{b}{\sqrt{2\pi}}\int_{-\infty}^\infty
e^{-\frac{1}{2}\zeta_1^2 -\frac{1}{2}(b\,\zeta_1-a)^2}\,\frac{d\zeta_1}{\sqrt{2\pi}} & \text{(integration by parts)} \\
&= \frac{b}{\sqrt{2\pi}}\,e^{-\frac{1}{2}\frac{a^2}{1+b^2}}\int_{-\infty}^\infty
e^{-\frac{1+b^2}{2}\left(\zeta_1-\frac{a\,b}{1+b^2}\right)^2}\,\frac{d\zeta_1}{\sqrt{2\pi}} & \text{(completing squares)} \\
&= \frac{b}{\sqrt{1+b^2}}\;\phi\left(\frac{a^2}{1+b^2}\right) \\
&= \sqrt{\frac{1-\rho}{2}}\;\phi\left(\frac{c}{\sqrt{2(1-\rho)}}\right)\,.
& \text{(inserting $a$ and $b$)}
\end{align*}
Thus, we use the above results to obtain
\[
A=\EE\left[\sigma\,Z_1\,\II_{\{Z_1+Z_2 > \frac{\xi}{\sigma}\}} \right]\,,
\]
and the previous result for $B$, and we arrive at \eqref{eqn: Omega ST on stale}.

\subsection{Proof of Proposition \ref{prop: Prob ST exec}}
\label{sec: Proof of Prob ST exec}

By conditioning on whether the trader receives the update or not, we have
\begin{align}
\PP[ P_e - P_2 > \xi]
&= \beta\, \PP\left[ P_e - P_2 > \xi \,| \,\text{update}\right] + (1-\beta) \, \PP\left[ P_e - P_2 > \xi \,| \, \text{stale}\right] \\
&= \beta\, \PP\left[ P_1 - P_2 > \xi \right] + (1-\beta) \, \PP\left[ P_0 - P_2 > \xi \right]\,,
\end{align}
and the result follows from computing these unconditional probabilities.

\subsection{Proof of Proposition \ref{thm: Losses to LA with last look}}
\label{sec: Proof of Losses to LA with last look}

To compute this we derive each term separately. Firstly,
\begin{align}
\Omega_{LA\,|\,buy}
=&\, \EE\left[ \; (P_1-P_0-\Delta)_+\,\II_{\{\,P_0-P_2> \xi \}}\; \right] \nonumber \\
=&\, \EE\left[ \; (P_1-P_0)\,\II_{\{\,P_1-P_0>\Delta\,,\,P_0-P_2> \xi\,\}}\; \right] \label{eqn: Omega LA term A}\\
& \quad - \Delta \; \PP\left[\;P_1-P_0>\Delta\,,\,P_0-P_2> \xi\,\right]\,.\label{eqn: Omega LA term B}
\end{align}
Let $\sigma \,B$ denote the first term above \eqref{eqn: Omega LA term A} and $\Delta\,A$ denote the second term above \eqref{eqn: Omega LA term B}.

First, focus on computing $A$, so we have
\begin{align*}
A &= \PP\left[\;P_1-P_0>\Delta\,,\,P_0-P_2> \xi\,\right] \\
&= \PP\left[\;Z_1>\tDelta\,,\,Z_1+Z_2< -\txi\,\right] \\
&= \PP\left[\;Z_1>\tDelta\,,\,\frac{(1+\rho)\,Z_1+\sqrt{1-\rho^2}\,Z_1^\perp}{\sqrt{2(1+\rho)}} <-\frac{\txi}{\sqrt{2(1+\rho)}}\;\right] \\
&= \PP\left[\;Z_1>\tDelta\,,\,Z_3 <- \frac{\txi}{\sqrt{2(1+\rho)}}\;\right] \\
&= \PP\left[\;Z_3 <- \frac{\txi}{\sqrt{2(1+\rho)}}\;\right] - \PP\left[\;Z_1<\tDelta\,,\,Z_3 <- \frac{\txi}{\sqrt{2(1+\rho)}}\;\right]\,,
\end{align*}
where $Z_1^\perp$ is a standard normal r.v. independent of $Z_1$, and $Z_3$ is standard normal r.v. correlated with $Z_1$ with correlation $\sqrt{(1+\rho)/2}$. The expression for $A$ in \eqref{eqn:LA Prob Exec I} now follows immediately.

Next, for $B$ we have
\begin{align*}
B =& \EE\left[Z_1\,\II_{\{Z_1>\tDelta\,,\,Z_2 < -\txi \}} \right] \\
=& \, \EE\left[Z_1\,\II_{\left\{Z_1>\tDelta\,,\,{(1+\rho)Z_1+\sqrt{1-\rho^2}Z_1^\perp} <- \txi\right\}} \right] \\
=& \int_{\tDelta}^\infty \zeta\,
\Phi\left(a-b\,\zeta\right) e^{-\frac{1}{2}\zeta^2} \, \frac{d\zeta}{\sqrt{2\pi}}\,,
\end{align*}
where $a=-\txi/\sqrt{1-\rho^2}$ and $b=(1+\rho)/\sqrt{1-\rho^2}$. Continuing the computation,
\begin{align*}
B=& -\frac{b}{\sqrt{2\pi}}\int_{\tDelta}^\infty
e^{-\frac{1}{2}(a-b\,\zeta)^2}\, e^{-\frac{1}{2}\zeta^2} \, \frac{d\zeta}{\sqrt{2\pi}} + \phi(\tDelta)\,\Phi(a-b\tDelta)
\qquad \qquad \text{   (integration by parts)}
 \\
=& -\frac{b}{\sqrt{2\pi}}\,e^{-\frac{1}{2}\frac{a^2}{1+b^2}}
\int_\tDelta^\infty e^{-\frac{1+b^2}{2}(\zeta-\frac{a\,b}{1+b^2})^2}\,\frac{d\zeta}{\sqrt{2\pi}}
+ \phi(\tDelta)\,\Phi(a-b\tDelta)
\qquad \text{(completing squares)} \\
=&\, - \frac{b}{\sqrt{1+b^2}} \,\phi\left(\frac{a^2}{1+b^2}\right)\,\Phi\left(\frac{a\,b}{\sqrt{1+b^2}}-\tDelta\sqrt{1+b^2}\right)
+ \phi(\tDelta)\,\Phi(a-b\tDelta)\,.
\end{align*}
The expression for $B$ in \eqref{eqn:LA Prob Exec II} follows by substituting the expression for $a$ and $b$ in the above. \qed

\subsection{Proof of Proposition \ref{prop: asymptotic spread}}
\label{sec: proof of prop: asymptotic spread}

First, from the expression for $\Omega_{ST}$ in \eqref{eqn: Omega ST exact} we have
\begin{equation}
\Omega_{ST}(\tDelta) = \sigma\left( a_{ST} + \tDelta\, b_{ST}\right)\,,
\label{eqn: Omega ST appendix}
\end{equation}
where
\[
a_{ST} =(1-\beta)\,\sqrt{\tfrac{1+\rho}{2}}\;\phi\left(\tfrac{1}{\sqrt{2(1+\rho)}}\,\tfrac{\xi}{\sigma}\right)\,,
\]
and
\[
b_{ST} =
\beta\,\Phi\left(-\tfrac{\xi}{\sigma}\right)+(1-\beta)
\Phi\left(-\tfrac{1}{\sqrt{2(1+\rho)}}\,\tfrac{\xi}{\sigma}\right)\,,
\]
and $\tDelta=\Delta/\sigma$.

Next, from  \eqref{eqn: Omega LA exact},
\[
\Omega_{LA}(\tDelta) = 2\,(B(\tDelta)- A(\tDelta)\,\tDelta)\,\sigma\,,
\]
where
\begin{align}
\nonumber A(\tDelta)&=\, a_{LA}
-\Phi_{\trho}\left(\tDelta\;,\;\hxi\right)\,,
\end{align}
with the constants
\[
a_{LA} = \Phi\left(-\tfrac{\txi}{\sqrt{2(1+\rho)}}\right)\,,
\quad
\hxi = -\tfrac{\txi}{\sqrt{2(1+\rho)}}\,,
\quad
\trho=\sqrt{\tfrac{1+\rho}{2}}\,,
\]
and
\begin{align}
\nonumber B(\tDelta)
&= \,
\phi(\tDelta)\,
\Phi\left(\cxi-\crho\,\tDelta\right)
-b_{LA} \,\Phi\left( \bxi -\brho\,\tDelta \right)\,,
\end{align}
with constants
\[
\cxi = -\tfrac{\txi}{\sqrt{1-\rho^2}}\,, \quad
\crho = \tfrac{1+\rho}{\sqrt{1-\rho^2}}\,, \quad
b_{LA} = \trho\,\phi\left(\tfrac{\txi}{\sqrt{2(1+\rho)}}\right)\,, \quad
\bxi = -\tfrac{\txi}{\sqrt{2(1-\rho)}}\,, \quad
\brho = \tfrac{2}{\sqrt{2(1-\rho)}}\,.
\]

The optimal spread (relative to volatility) is defined as the solution to
\[
(1-\alpha)\, \Omega_{ST}(\tDelta^*) - \alpha\,\Omega_{LA}(\tDelta^*) =0\,,
\]
and writing $\tDelta^*=\tDelta_0+\tDelta_1\,\alpha+o(\alpha)$, we need to solve (keeping terms to $o(\alpha)$):
\[
(1-\alpha) \, \left(\Omega_{ST}(\tDelta_0)+\alpha\,\tDelta_1\,\Omega_{ST}'(\tDelta_0) \right) - \alpha\,\Omega_{LA}(\tDelta_0) = o(\alpha)\,,
\]
and collecting terms of equal orders we have
\[
\Omega_{ST}(\tDelta_0)
+ \alpha \left\{
\tDelta_1\,\Omega_{ST}'(\tDelta_0)
-\Omega_{ST}(\tDelta_0)
-\Omega_{LA}(\Delta_0) \right\}
=o(\alpha)\,.
\]
Solving first for $\tDelta_0$ by setting the first term above to zero and using \eqref{eqn: Omega ST appendix}, we arrive at  $\tDelta_0 = \frac{a_{ST}}{b_{ST}}$ and  \eqref{eqn: asymp spread Delta_0} follows.

Hence, the above equation becomes
\[
\alpha \left\{
\tDelta_1\,\Omega_{ST}'(\tDelta_0)
+\Omega_{LA}(\Delta_0) \right\}
=o(\alpha)\,.
\]
Finally, since $\Omega_{ST}'(\tDelta) = b_{ST}$, setting the terms in the braces to zero leads to $\tDelta_1 = \frac{\Omega_{LA}(\tDelta_0)}{b_{ST}}$ and \eqref{eqn: asymp spread Delta_1} follows immediately. \qed

\subsection{Proof of Proposition \ref{prop: asymptotic omega ft}}
\label{sec: proof of prop: asymptotic omega ft}

From the proof of Proposition \ref{prop: asymptotic spread} in \ref{sec: proof of prop: asymptotic spread}, and since $\Psi_{ST}$ is independent of $\alpha$ and therefore $\Delta$, when the broker sets spreads at their optimal level according to \eqref{eqn: optimal spread balancing equation}, we immediately have that
\begin{align*}
\widehat\Omega_{ST} =&\, \Omega_{ST}(\tDelta^*)+ \delta\,(1-\Psi_{ST}) \\
=&\, \sigma(a_{ST} + b_{ST} (\tDelta_0+\alpha\,\tDelta_1))+ \delta\,(1-\Psi_{ST}) +o(\alpha) \\
=&\, \sigma\ \alpha\,b_{ST}\,\tDelta_1+ \delta\,(1-\Psi_{ST}) +o(\alpha)\,,
\end{align*}
and using $\tDelta_1 = \frac{\Omega_{LA}(\tDelta_0)}{b_{ST}}$, the proof is complete. \qed

\subsection{Proof of Proposition \ref{prop: asymptotic omega la}}
\label{sec: proof of prop: asymptotic omega la}

Using the same notation as in the proof of Proposition \ref{prop: asymptotic spread} in \ref{sec: proof of prop: asymptotic spread}, we have
\[
\Omega_{LA}(\tDelta) = 2\,\left(B(\tDelta)- A(\tDelta)\,\tDelta\right)\,\sigma\,,
\]
where
\begin{align}
\nonumber A(\tDelta)&=\, a_{LA}
-\Phi_{\trho}\left(\tDelta\;,\;\hxi\right)\,,
\end{align}
and
\begin{align}
\nonumber B(\tDelta)
&= \,
\phi(\tDelta)\,
\Phi\left(\cxi-\crho\,\tDelta\right)
-b_{LA} \,\Phi\left( \bxi -\brho\,\tDelta \right)\,.
\end{align}
When the broker sets spreads at their optimal level so that her expected profit and loss is zero, we have
\begin{align*}
\frac{1}{2\sigma}\Omega_{LA}(\tDelta^*)
=& \,
\left\{B(\tDelta_0+\alpha\,\tDelta_1)- A(\tDelta_0+\alpha\,\tDelta_1)\,(\tDelta_0+\alpha\,\tDelta_1)
\right\} + o(\alpha) \\
=& \,
\left\{B(\tDelta_0)- A(\tDelta_0)\,\tDelta_0\right\}
+ \left\{ B'(\tDelta_0) - A(\tDelta_0) - A'(\tDelta_0)\,\tDelta_0\right\} \alpha + o(\alpha)\,.
\end{align*}
Next, by direct computations
\[
B'(x) = -x\,\phi(x)\,\Phi\left(\cxi-\crho\, x\right) - \crho \,\phi(x)\,\phi\left(\cxi-\crho x\right)+\brho\,b_{LA}\,\Phi\left(\bxi-\brho\, x\right)\,,
\]
and
\[
A'(x) = -\partial_1\Phi_{\trho}\left(\tDelta,\hxi\right)\,.
\]
Further, standard computations show that
\begin{align*}
\partial_1\Phi_{c}(x,y)
=&\, \int_{-\infty}^y\phi_c(x,y)\,dx\,dy \\
=&\, \int_{-\infty}^y \exp\left\{-\tfrac{1}{2(1-c^2)}\left(x^2 - 2c\,x\,u + u^2 \right)\right \}\frac{dx\,dy}{2\pi} \\
=&\, \frac{1}{\sqrt{2\pi}} e^{-\frac{1}{2}x^2} \int_{-\infty}^{\frac{y-c\,x}{\sqrt{1-c^2}}} e^{-\frac{1}{2}z^2}\frac{dz}{\sqrt{2\pi}} \\
=&\, \sqrt{1-c^2}\,\phi(x)\,\Phi\left(\frac{y-c\,x}{\sqrt{1-\rho^2}}\right)\,.
\end{align*}
Inserting the explicit expressions for the various constants completes the proof. \qed

\clearpage

\section*{References}
\bibliographystyle{chicago}
\bibliography{LastLookRefs}

\begin{thebibliography}{}

\bibitem[\protect\citeauthoryear{{Bank of England, H.M. Treasury, and Financial
  Conduct Authority}}{{Bank of England, H.M. Treasury, and Financial Conduct
  Authority}}{2015}]{BoE}
{Bank of England, H.M. Treasury, and Financial Conduct Authority} (2015, June).
\newblock How fair and effective are the fixed income, foreign exchange and
  commodities markets?
\newblock Fair and Effective Markets Review.

\bibitem[\protect\citeauthoryear{Copeland and Galai}{Copeland and
  Galai}{1983}]{copeland1983information}
Copeland, T.~E. and D.~Galai (1983).
\newblock Information effects on the bid-ask spread.
\newblock {\em The Journal of Finance\/}~{\em 38\/}(5), 1457--1469.

\bibitem[\protect\citeauthoryear{de~Jong and Rindi}{de~Jong and
  Rindi}{2009}]{FrankDeJong2009}
de~Jong, F. and B.~Rindi (2009).
\newblock {\em The Microstructure of Financial Markets\/} (1st ed.).
\newblock Cambridge University Press.

\bibitem[\protect\citeauthoryear{Glosten and Milgrom}{Glosten and
  Milgrom}{1985}]{glostenMilgrom}
Glosten, L.~R. and P.~R. Milgrom (1985).
\newblock Bid, ask and transaction prices in a specialist market with
  heterogeneously informed traders.
\newblock {\em Journal of Financial Economics\/}~{\em 14\/}(1), 71--100.

\bibitem[\protect\citeauthoryear{Grossman and Miller}{Grossman and
  Miller}{1988}]{GrossmanMiller1988}
Grossman, S.~J. and M.~H. Miller (1988, July).
\newblock Liquidity and market structure.
\newblock {\em The Journal of Finance\/}~{\em 43\/}(3), 617--37.

\bibitem[\protect\citeauthoryear{Oomen}{Oomen}{2017a}]{OomenAgr}
Oomen, R. (2017a).
\newblock Execution in an aggregator.
\newblock {\em Quantitative Finance\/}~{\em 17\/}(3), 383--404.

\bibitem[\protect\citeauthoryear{Oomen}{Oomen}{2017b}]{OomenLastLook}
Oomen, R. (2017b).
\newblock Last look.
\newblock {\em Quantitative Finance\/}~{\em 17\/}(7), 1057--1070.

\end{thebibliography}
\end{document}